\documentclass{article}

\usepackage{arxiv}

\usepackage[utf8]{inputenc} 
\usepackage[T1]{fontenc}    
\usepackage{amsmath,amsfonts,amsthm,amssymb}

\usepackage[square,numbers]{natbib}

\usepackage{mathtools}
\usepackage{setspace}
\usepackage{fancyhdr}
\usepackage{makecell}
\usepackage{lastpage}
\usepackage{extramarks}
\usepackage{hyperref}
\usepackage{xhfill}
\usepackage{chngpage}
\usepackage{soul,color}
\usepackage{graphicx,float,wrapfig}
\usepackage{booktabs}
\usepackage[numbers]{natbib}
\usepackage{bbm}
\usepackage{bm}
\usepackage{graphicx}
\usepackage{enumerate}
\usepackage{bbm}
\usepackage{pgf}
\usepackage{tikz}

\newtheorem{assumption}{Assumption}
\newtheorem{proposition}{Proposition}
\newtheorem{theorem}{Theorem}
\newtheorem{lemma}{Lemma}
\newtheorem{corollary}{Corollary}
\newtheorem{definition}{Definition}
\newtheorem{remark}{Remark}
\newcommand{\set}[1]{\mathcal{#1}}
\newcommand{\expect}[1]{\mathbb E\left[#1\right]}
\renewcommand{\Pr}{\mathbb{P}}

\newcommand{\thmnamerestated}{placeholder}
\allowdisplaybreaks

\begin{document}

\title{Optimal Load Balancing in Bipartite Graphs}


\author{
  Wentao Weng\\
  Institute for Interdisciplinary Information Sciences \\
  Tsinghua University\\
  \texttt{wwt17@mails.tsinghua.edu.cn} \\
  \And
  Xingyu Zhou\\
  ECE \\
  Ohio State University\\
  \texttt{zhou.2055@osu.edu} \\
  \And
  R. Srikant \\
  C3.ai DTI, CSL and ECE \\
  University of Illinois at Urbana-Champaign \\
  \texttt{rsrikant@illinois.edu}
}
\date{}

\maketitle

\begin{abstract}
Applications in cloud platforms motivate the study of efficient load balancing under job-server constraints and server heterogeneity. In this paper, we study load balancing on a bipartite graph where left nodes correspond to job types and right nodes correspond to servers, with each edge indicating that a job type can be served by a server. Thus edges represent locality constraints, i.e., each job can only be served at servers which contained certain data and/or machine learning (ML) models. Servers in this system can have heterogeneous service rates. In this setting, we investigate the performance of two policies named Join-the-Fastest-of-the-Shortest-Queue (JFSQ) and Join-the-Fastest-of-the-Idle-Queue (JFIQ), which are simple variants of Join-the-Shortest-Queue and Join-the-Idle-Queue, where ties are broken in favor of the fastest servers. Under a ``well-connected'' graph condition, we show that JFSQ and JFIQ are asymptotically optimal in the mean response time when the number of servers goes to infinity. In addition to asymptotic optimality, we also obtain upper bounds on the mean response time for finite-size systems. We further show that the well-connectedness condition can be satisfied by a random bipartite graph construction with relatively sparse connectivity.
\end{abstract}

\section{Introduction}
Many applications that use data centers, cloud computing systems and other data analytic platforms, including Web search engines \cite{Google-search}, cloud computing service \cite{Amazon-aws}, large-scale data processing \cite{dean2008mapreduce}, and cloud storage have extremely stringent latency requirements. Ultra low latency guarantees in these applications not only provide smooth user experience, but help improve company profits \cite{dean2013}. 

A key component for achieving a fast response in the aforementioned systems are load balancing algorithms, which are responsible for dispatching jobs to parallel servers. Motivated by the demanding requirement of a low latency, there has been a line of recent research that aims to design smart load balancing algorithms with delay performance guarantees. They often focus on the classical load balancing model, where there are $N$ identical servers with exponential service times and a dispatcher that assigns Poisson arrivals to one of the servers. It has been shown that in this setting that a class of load balancing policies including Join-the-Shortest-Queue (JSQ), Join-the-Idle-Queue (JIQ) \cite{lu2011join} and variants of the Power-of-d-Choices (Pod) \cite{mitzenmacher2001power,vvedenskaya1996queueing} which sample a sufficiently large number of queues or exploit the parallelism of tasks within a job are able to achieve asymptotically zero waiting time for a sufficiently large $N$. 

However, the above classical load balancing model may not be appropriate for certain modern cloud computing and data analytic applications due to the presence of job-server constraints. Under such constraints, a job can only be dispatched to a subset of the $N$ servers. These constraints, often called locality constraints, are quite common in large-scale Machine Learning as a Service (MLaaS) and serverless computing services supported by cloud computing platforms (e.g., Microsoft Azure \cite{Microsoft-azure}, Amazon Web Services \cite{Amazon-aws}, Google Cloud \cite{Google-cloud}). To give a concrete example, let us consider MLaaS. In this setting, various well-trained machine learning models are deployed on cloud platforms, say deep convolutional neural network (CNN) models for image classification and natural language processing (NLP) models. A user’s image classification request can only be sent to the servers on which the CNN models have been loaded. As a result, it is not appropriate to assume that every request can be served by any server in the system. Other examples in which there are inherent job-server constraints include online video services, such as TikTok, Netflix and Youtube. In these applications, user requests can only be sent to servers with the required data (e.g., movies, music). The ultimate goal in all these modern applications is to achieve a fast response time and efficient resource (e.g., number of servers) usage  while satisfying job-server constraints.

Inspired by these applications, in this paper, we take into account job-server constraints by considering a bipartite load balancing model. In this model, job-server constraints are abstracted by the edges in a bipartite graph, where the left nodes are called ports and the right nodes are called servers. In the model, each port represents a job of a particular type which requires a specific chunk of data or a specific machine learning model to execute, and thus can only be routed to specific servers. Each port $\ell$ corresponds to Poisson job arrivals with rate $\lambda_{\ell}$. A job from a port $\ell$ can only be sent to server $r$ such that $(\ell,r)$ is an edge of the graph. Jobs routed to a server $r$ are queued in a buffer, and get service in a first-come first-server manner. The service time of each job at server $r$ is exponentially distributed with rate $\mu_r$ (possibly different). 

To the best of our knowledge, this bipartite graph model was only introduced recently in \cite{cruise2020stability}, where JSQ is shown to be throughput optimal while no delay performance guarantee is provided. The bipartite graph model generalizes the load balancing model on graphs introduced in \cite{mukherjee2018asymptotically,budhiraja2019supermarket}. In their model, jobs arrive at each node with a homogeneous rate, and each job can be served by the node it arrives and its neighbors. It has been shown that in this setting JSQ achieves zero delays under certain assumptions on graph connectivity \cite{mukherjee2018asymptotically}. 

Inspired by the discussions above,  we are particularly interested in the following question: 

\textit{
Are there simple policies that can achieve optimal response time in modern load balancing systems with both job-server constraints and service-rate heterogeneity?
}  

\subsection{Main Contribution} 
This paper affirmatively answers the above question by presenting optimal policies as well as performance bounds on the mean response time. The detailed contributions can be summarized as follows.

First, we consider two policies: Join-the-Fastest-of-the-Shortest-Queues (JFSQ), and Join-the-Fastest-of-the-Idle-Queues (JFIQ). 
We show that, under a ‘well-connected’ graph condition,  they can asymptotically achieve the minimum response time in both the many-server regime (the system load $\lambda < 1$ is a constant while the number of servers $N \to \infty$) and sub Halfin-Whitt (HW) regime ($\lambda = 1 - N^{-\alpha}$ with $\alpha < 0.5$). The minimum response time metric is more stringent than the common "zero queueing delays" discussed before, and is especially important in systems with heterogeneous servers. JFSQ and JFIQ are simple variants of JSQ and JIQ adapted to job-server constraints, but they break ties in JSQ and JIQ by choosing the fastest servers. Consequently, our results imply that JSQ and JIQ have asymptotic zero waiting time for homogeneous servers. They are practical since they only need comparisons between service speed rather than the exact service rates of servers. In addition to the asymptotic result, we also obtained finite-system bounds on the mean response time. Roughly speaking, we show that the difference between the mean response time in an $N$-server system and that in the limit is bounded by $O\left(\epsilon+((1-\lambda)\epsilon N)^{-1/2}\right)$, where $\epsilon$ is a parameter related to the well-connectedness of the underlying bipartite graph, and $\lambda$ reflects the load of the system.

Second, our theoretical results provide practical guidance in designing modern load balancing systems. Besides the two simple but efficient algorithms, the underlying  ‘well-connected’ condition sheds light on the efficient deployment of various ML models or the required data among the servers. In particular, the key message is that each movie on Netflix or each ML model deployed on Microsoft Azure only needs to be loaded in $\omega(1)$ servers. To give a concrete example, we show that if edges in the bipartite graph are randomly generated according to some given probabilities, then the graph is "well-connected" with high probability. Let $L$ be the number of kinds of jobs, and $N$ be the number of servers. Our result indicates that on average, the graph only needs $\omega\left(\frac{L+N}{(1-\lambda)^2}\right)$ connections to be "well-connected". And if the arrival rates of jobs are uniform, then this number can be reduced to $\omega\left(\frac{L+N}{1-\lambda}\ln\frac{1}{1-\lambda}\right)$.

A key theoretical contribution of the paper is showing that a recently-developed Lyapunov drift method for studying parallel-server queueing systems can be generalized to bipartite graphs using two key ideas: (i) we demonstrate something akin to state-space collapse and resource pooling by exploiting the connectivity structure of the graph, and (ii) apply this idea iteratively twice, once to bound the number of jobs in fast servers that are busy in the large-system limit and a second time to bound the number of jobs in slow servers that are idle in the limit using a conditional geometric tail bound.

\subsection{Related Work}
There is a vast literature on efficient load balancing policies, mostly in the classical load balancing setting where there are $N$ identical servers and the service rate is exponentially distributed. Upon arrival, each job can be sent to any of the $N$ servers. It is now well-known that in this setting JSQ is optimal \cite{Weber78} in a stochastic ordering sense. However, obtaining the exact steady state performance of JSQ is difficult. The problem is partly solved in \cite{eschenfeldt2018} which establishes that the scaled queue length process of JSQ converges to a two-dimensional Ornstein-Uhlenbeck process, and the fraction of waiting jobs vanishes in the Halfin-Whitt heavy traffic regime. Although this result is on the process level, it is later confirmed for the steady state distribution by \cite{braverman2020}. The tail of the distribution is further studied in \cite{banerjee2019join}. 

Since JSQ has significant communication overhead in large-scale systems, alternative policies have been proposed and analyzed. One prominent policy is Power-of-$d$-Choices (Pod). In Pod, each arrival of jobs probes $d$ random servers, and joins the one with the shortest queue. \cite{mukherjee2018} first shows that if $d \to \infty$, then both the fluid limit and the state occupancy distribution of Pod coincides with that of JSQ in many-server limit. It implies that Pod has zero waiting time in many-server limit. \cite{mukherjee2018} also prove that the diffusion limit of Pod is the same as JSQ if $d = \omega(\sqrt{N}\log N)$ in the Halfin-Whitt heavy traffic regime, but it does not induce steady-state performance. For the many-server regime, a line of works \cite{gamarnik_2018,gamarnik2020lower} study the minimum required resources (such as memory, and communication overhead) to achieve zero waiting time. 

When the system load $\lambda$ can also approach $1$ as $N$ increases (i.e. many-server heavy-traffic regime), \cite{liu2018achieving} shows that Pod can achieve asymptotic zero waiting time if $d = \omega\left(\frac{1}{1-\lambda}\right)$ when $1-\lambda = \omega(N^{-1/6})$. For a heavier-traffic regime, a recent breakthrough is the work \cite{LiuYing2020}. In the sub Halfin-Whitt regime ($1-\lambda=\omega(N^{-0.5})$), this work establishes asymptotic zero waiting property for a large class of policies including JSQ, JIQ and Pod with $d =O(\frac{\log N}{1-\lambda})$. The result is later extended to the Beyond-Halfin-Whitt regime ( $1-\lambda=\omega(N^{-1})$) \cite{LiuYing2019-universal}, and to Coxian-2 service time distribution \cite{LiuYing2020-Coxian}. When $1-\lambda=O(N^{-1})$, it is known that the waiting time must be positive for all load balancing policies \cite{atar2012diffusion,gupta2019load}. When jobs are divisible, \cite{WengWang20, mukherjee2018} shows similar result for Batch Sampling \cite{ousterhout2013sparrow} and Batch-Filling \cite{Ying2017}, which are batch variants of Pod.  

Proving optimality of load balancing algorithms is more complicated when servers are heterogeneous. Simple heuristics, nevertheless, are proposed in decades. We note that a policy called \textit{Never Queue} policy which is very similar to JFIQ was proposed in \cite{shenker1989optimal}. The Never Queue policy is analyzed in the case of a centralized queue, but not for load balancing systems. Many studies have focused on the heavy traffic regime where the system load converges to $1$ while the number of servers is fixed. In this regime, JSQ was shown to be delay optimal by the drift method \cite{EryilmazSrikant2012}. Later, \cite{zhou2018heavy} proves that a threshold policy is heavy-traffic optimal. The stability and optimality in heavy traffic of Pod  for heterogeneous servers studied recently by \cite{hurtado2020throughput}. Moreover, \cite{zhou2018flexible} provides a simple criteria for load balancing algorithms to be heavy-traffic optimal. The assumption of heavy traffic can be relaxed to many-server heavy traffic regime when $1-\lambda = o(N^{-4})$ \cite{hurtado2020load,zhou2020note}. Nevertheless, the results mentioned above do not imply fast mean response time in the many-server regime, which is more practical for cloud platforms. For the many-server regime, work in \cite{stolyar_2015} shows that JIQ has asymptotic zero waiting time as $N \to \infty$. However, this does not imply optimal mean response time since the service time of jobs varies in different servers. A recent work \cite{gardner2020scalable} takes heterogeneity into accounts by studying a system with fast and slow servers. Although \cite{gardner2020scalable} obtains mean-field limit for a variant policy of Pod, the result does not imply optimal mean response time. 

Load balancing with job-server constraints are not considered in the literature until recent years. To the best of our knowledge, \cite{moharir2015online} is the first paper that considers load balancing with job-server constraints and proposes an online load balancing algorithm with the optimal competitive ratio. However, their model is not stochastic, and is thus quite different from the model we are considering in this paper. \citet{cruise2020stability} considers the stability of JSQ on the same model as ours while no delay guarantee is provided. In \citet{cardinaels2020redundancy}, redundancy policies are explored in bipartite load balancing. They obtain a product-form steady state distribution which however does not imply an optimal mean response time. Besides these papers, there are also studies for load balancing on graphs. In \cite{turner1998effect,gast2015power,budhiraja2019supermarket}, the impact of the graph structure on the performance of Pod is studied. \citet{mukherjee2018asymptotically} utilizes a stochastic coupling method to prove that JSQ on graph can have the same performance as JSQ in the classical load balancing model in both the many-server regime and the Halfin-Whitt regime under certain graph constraints. Therefore, it implies that JSQ can also achieve zero waiting time in the many-server regime for a graph-based model. However, the model in \cite{mukherjee2018asymptotically} only considers identical servers and homogeneous arrival rates of jobs, which is a special case of this paper. 

We note that if servers share a central queue, then the bipartite graph model turns into the skill-based model studied in the call center literature \cite{gardner2020product, cardinaels2020redundancy}. It is shown in \cite{gardner2020product} (and the references within) that the stationary distributions under several redundancy policies have product forms. One related result to us is that our model becomes the same as a skill-based model, and thus enjoys a product-form stationary distribution, if we send a job to a connected server with least amount of work in its buffer \cite{gardner2020product,cardinaels2020redundancy}. Such policy is, however, impractical since workloads of jobs in cloud platforms suffer from volatility. Also, as \cite{gardner2020product} has pointed out, it is non-trivial to obtain bounds on mean response time just from the product-form results. 

Our bipartite graph model also resembles other problems in the literature. One particular model is the job-server affinity model for data locality problems studied in \cite{cardinaels2019job,xie2015priority,xie2016scheduling,wang2014maptask}. In the job-server affinity model, if one job is served by a server with its data, it has a fast constant service rate. Otherwise, it has a slow service rate, meaning that this sever has to fetch data from somewhere. However, the setting is not suitable in the context of MLaaS we discussed above. Here ML models are usually reconfigured on machines periodically, and a new request will only be routed to those servers with needed model \cite{gujarati2017}. Also, previous studies on job-server affinity models can only guarantee heavy-traffic delay optimality \cite{xie2015priority,xie2016scheduling,wang2014maptask}, which does not induce extremely fast mean response time required in cloud platforms. 

From a methodological perspective, our paper builds on the drift method to obtain performance bounds. In this method, one exploits the fact that the steady-state expectation of suitable functions of the state of a Markov process does not change with time. This idea was developed in~\cite{EryilmazSrikant2012,maguluri2016heavy,WangMaguluri18} for the heavy-traffic regime where the idea of using the tail bounds of \cite{hajek1982hitting,BertsimasGarmarnik01} to prove state-state collapse or resource pooling was introduced. The recent work in \cite{LiuYing2020} developed a parallel approach for the many-server regime where they introduced the notion of generator coupling inspired by Stein's method in \cite{ying2017stein,braverman2017stein,gurvich2014diffusion,stolyar2015tightness} and designed a clever Lyapunov coupling to show that, for JSQ-type policies, the number of homogeneous servers utilized is large when the backlog is large. We will call this latter idea \emph{state-space collapse} since it is similar to the notion of state-space collapse in the heavy-traffic regime. In this paper, we introduce new ideas to expand the applicability of the techniques \cite{LiuYing2020} to networks of heterogeneous servers. 

Contemporaneous to our work, in \cite{Rutten2020}, the authors study the waiting time of JSQ(d) policies in bipartite graphs in the limit as the size of the graph goes to infinity. While the papers are motivated by related problems, the models and routing policies studied, and the results in the two papers are different. The authors in \cite{Rutten2020} consider the case of homogeneous servers with infinite buffers, and show that the performance of JSQ(d) in a bipartite graph with limited connectivity converges to the performance of the fully flexible system in terms of queue length (or waiting time) under appropriate connectivity conditions. In addition, they prove that the occupancy in steady state of the limited-connectivity system converges to the steady state of the fully flexible system. Our paper considers the case of heterogenous arrival and service rates with finite buffers, and shows that the waiting time in the queue and blocking probability both go to zero in the large-system limit under the JFIQ and JFSQ routing policies. Additionally, the techniques used in the two papers are different. We use the drift method to obtain performance bounds for finite-sized systems while \cite{Rutten2020} uses process-level convergence techniques.

\begin{figure}
    \centering
    \includegraphics[scale=0.25]{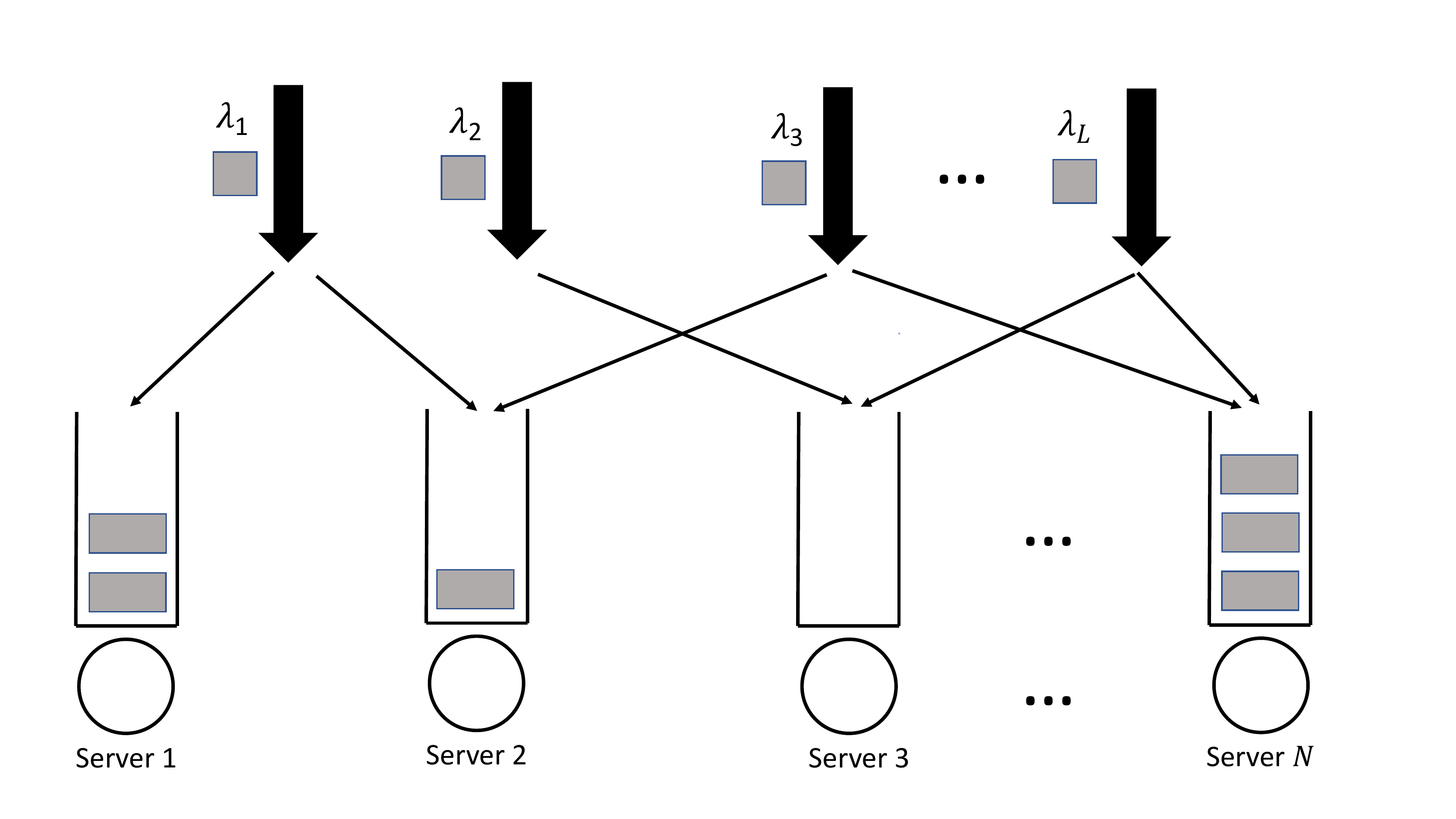}
    \caption{An example of the bipartite graph model. In this instance, jobs from port $1$ can only be routed to server $1$ and server $2$.}
    \label{fig:model}
\end{figure}
\section{Model}
We consider load balancing in a bipartite graph $G=(\set{L},\set{R},E)$ where $\set{L}$ and $\set{R}$ are the set of left nodes and right nodes, respectively, and $E$ is the set of edges between these two sets of nodes. Nodes in $\set{L}$ are indexed as $\{1,2,\cdots,L\}$ with $L = |\set{L}|$, and nodes in $\set{R}$ are indexed as $\{1,2,\cdots,N\}$ with $N = |\set{R}|$. For a node $\ell \in \set{L}$ (or $r \in \set{R}$), define $\set{N}_{L}(\ell)$ (or $\set{N}_{R}(r)$) to be the set of right (or left) nodes it connects with. W.L.O.G., every $\set{N}_{L}(\ell),\set{N}_{R}(r)$ is assumed to be non-empty. To distinguish between left and right nodes, we may refer to a node $\ell \in \set{L}$ as port $\ell$, and a node $r \in \set{R}$ as server $r$. See Fig. \ref{fig:model} for an illustration. 

Jobs arrive at port $\ell$ according to a Poisson process with rate $\lambda_{\ell}$, and the goal is to route them to one of the servers connected to $\ell$ so as to minimize a certain performance metric of interest. It is assumed that every server has a finite buffer of size $b$. When a job is routed to a server that is currently processing another job, this new arrival will be placed in the buffer. But if there are already $b$ jobs (including the one being served), the new arrival is blocked and lost forever. We assume that jobs in the buffer are served in a first-come-first-serve manner. The queue length $Q_r$ of a server $r$ is the number of jobs in the buffer plus one if there is a job running on the server.

To reflect the nature of server heterogeneity in a practical load balancing system, we assume that there are $M$ types of servers. For a type $m$ server, the service time of a job running on it is assumed to be exponentially distributed with mean $\frac{1}{\mu_m}$. The arrival processes to the ports and the service times of jobs are assumed to be independent. Denote the number of type $m$ servers by $N_m$, and the type of a server $r$ by $t_r$. Equivalently, we can write $N_m = N\alpha_m$ with $\alpha_m \in (0,1), \sum_{m=1}^M \alpha_m = 1$. We assume that there is sufficient service capacity, i.e.,  $\lambda_{\Sigma} = \sum_{\ell=1}^L \lambda_{\ell} < N\sum_{m=1}^M \mu_m\alpha_m.$ W.L.O.G., we assume $\mu_1 > \mu_2 > \cdots > \mu_M > 0$ since we can always reorder the types of servers.  

We study two routing policies, Join-the-Fastest-of-the-Shortest-Queues (JFSQ) and Join-the-Fastest-of-the-Idle-Queues (JFIQ) in bipartite load balancing systems. For JFSQ, upon the arrival of a job at port $\ell$, we select a server $r$ connected to port $\ell$ with the shortest queue length, that is, $r \in \arg\min_{r \in \set{N}_L(\ell)} Q_r$. If there are multiple such servers, we select the one with the fastest service rate, i.e. largest $\mu_{t_r}$, and break ties (if any) by randomly choosing one server. Alternatively, if we use JFIQ, we find an idle server $r \in \set{N}_L(\ell)$ with the fastest service rate. If there is no idle servers, we select one server from $\set{N}_L(\ell)$ randomly. The question of interest in this paper is whether these two policies can achieve optimal job delays (at least for a large system) under appropriate conditions on the underlying bipartite graph. We note that our routing policies JFIQ and JFSQ reduce to JIQ and JSQ, respectively, when all servers have the same service rates.

\subsection{State Representation}
Before we proceed to state our results, we first state the notation that we will use in the paper. We use capital letters to denote random variables, such as $Q_r(t)$ for the queue length of server $r$ at time $t$, and small letters to denote realizations. 

Clearly, for the system considered in this paper, the sequence $\{\mathbf{Q}(t) = (Q_1(t),\cdots,Q_N(t))\}$ forms a Continuous Time Markov chain (CTMC). Since the buffers are finite, there is a unique stationary distribution of $\mathbf{Q}(t)$. For each state $\mathbf{q} = (q_1,\cdots,q_N)$, let 
\[
s_{m,i}(\mathbf{q}) = \frac{1}{N}\left|\left\{r \in \set{R}\colon q_r \geq i, t_r = m\right\}\right|
\]
be the fraction of type $m$ servers with queue length at least $i$. Besides, let 
\[
C_m(\mathbf{q}) = \sum_{i=1}^b s_{m,i}(\mathbf{q}), W(\mathbf{q}) = \sum_{m=1}^K \mu_m s_{m,1}(\mathbf{q}),
\]
which is the normalized (divided by $N$) number of jobs in type $m$ servers, and the rate to complete a job if we only consider the first $K$ types of servers. 
\paragraph{Notation:} As mentioned earlier, capital letters are reserved for random variables (such as $\mathbf{Q}(t)$ for queue lengths at time $t$), and small letters are for realizations (such as $\mathbf{q}$ for a queue-length state). We add a line on top of a variable meaning that it is in steady state (such as $\bar{\mathbf{Q}}$). This paper makes use of asymptotic notations. For two positive functions $f(x),g(x)$, we write $f(x)=o(g(x))$ if $\sup \lim_{x \to \infty} \frac{f(x)}{g(x)} = 0$; write $f(x)=O(g(x))$ if $\sup \lim_{x \to \infty} \frac{f(x)}{g(x)} < \infty$; write $f(x) = \Omega(g(x))$ if $\inf \lim_{x \to \infty} \frac{f(x)}{g(x)}>0$; write $f(x) = \omega(g(x))$ if $\inf \lim_{x \to \infty} \frac{f(x)}{g(x)}=\infty$.

\section{Main Results}
We summarize our main results in this section. To be specific, our results provide an upper bound of the mean number jobs in the system under certain assumptions. This upper bound can directly imply asymptotic optimality of JFSQ and JFIQ in the sense of minimum mean response time, which we will define explicitly later. We also give a random graph construction of the graph $G$ such that $G$ can satisfy Assumption \ref{as:dense} with high probability.

\subsection{Upper Bound of the Mean Number of Jobs}

Let $K$ be the minimum value such that $N\sum_{m=1}^K \mu_m\alpha_m > \lambda_{\Sigma}$. Such a $K$ must exist by the assumption of sufficient service capacity. Assume that $\lambda_{\Sigma} = N\sum_{m=1}^K\mu_m\alpha_m(1-\beta)$ where $0 <\beta \leq 1$, and denote $\lambda = \frac{\lambda_{\Sigma}}{N}$. Let
\[
C_1^* = \alpha_1,\cdots,C_{K-1}^*=\alpha_{K-1},C_K^*=\frac{\lambda-\sum_{m=1}^{K-1}\mu_m\alpha_m}{\mu_K},
\]
and let $C^* = \sum_{m=1}^K C^*_m$. Such definition is motivated by the mean-field limit of our system, which will be illustrated later. The following result provides lower bounds for the expected service time of each job, and the mean number of jobs in the system. 
\begin{proposition}\label{thm:lower-bound}
Suppose that the buffer size is infinite, i.e. $b = \infty$. Let $\bar{Z}$ be the random variable denoting the service time of one job. Then for any stable policy, the mean number of jobs in the system is lower bounded by $NC^*$, and 
\begin{equation}
\expect{\bar{Z}} \geq \frac{C^*}{\lambda}.
\end{equation}
\end{proposition}
The proof is provided in the appendix.

For every $1 \leq m \leq K$, let $\set{R}_m$ be the set of servers of types $1$ through $m$. Let $\hat{\beta} = \beta \sum_{m=1}^K \alpha_m$, and $\epsilon$ be a number in $(0,\frac{\hat{\beta}}{4}];$ we call $\epsilon$ the approximation error since we will later use this parameter to characterize the near optimality of our routing policies. For any subset $\set{I} \subseteq \set{R}$, define $\set{N}_{\set{R}}(\set{I}) = \cup_{r \in \set{I}} \set{N}_{\set{R}}(r)$ to be the set of ports connected to at least one server in $\set{I}$, and $D_{\set{I}} = \sum_{\ell \not \in \set{N}_{\set{R}}(\set{I})} \lambda_{\ell}$ be the sum of arrival rates at ports not connected to $\set{I}$. 
Before stating our results on JFSQ and JFIQ, we first make a few assumptions on the system. Let $\tau_{1K}=\frac{\mu_1}{\mu_K}, \tau_{1M}=\frac{\mu_1}{\mu_M}, \tau_{KM}=\frac{\mu_K}{\mu_M}$.

\begin{assumption}[Buffer Size]\label{as:buffer}
For a fixed approximation parameter $\epsilon$ in $(0,\frac{\hat{\beta}}{4}],$ the buffer size $b$ satisfies
$
6\sqrt{\tau_{1K}} \leq b \leq \left\lfloor \left(\frac{\epsilon^2 N}{1152\tau_{1K}\ln{N}}\right)^{1/5}\right \rfloor.
$
\end{assumption}
\begin{assumption}[Well Connectedness]\label{as:dense}
The graph $G$ satisfies the following conditions:
\begin{itemize}
\item $D_{\set{I}} \leq N\tilde{d_1}$ for any $\set{I} \subseteq \set{R}_{K-1}$ with $|\set{I}| \geq Np_1$;
\item $ D_{\set{I}} \leq N\tilde{d}_2$ for any $\set{I} \subseteq \set{R}_K$ with $|\set{I}| \geq Np_2$.  
\end{itemize}
where $p_1 = \frac{\epsilon}{6b^2}, p_2 = \frac{\hat{\beta}}{2}, \tilde{d}_1\leq\frac{\epsilon\mu_K}{12b^3}, \tilde{d}_2 \leq \frac{\epsilon\mu_K}{2b}$.
\end{assumption}
Although there are two constraints, Assumption \ref{as:dense} basically requires that a large enough subset of the first $K$ types of servers must connect with ports with enough arrival rates. Such requirement enables that JFSQ and JFIQ behave almost the same as in a classical load balancing system even though there are additional job-server constraints. We are now ready to state the main result.
\begin{theorem}\label{thm:finite-bound}
Suppose that Assumptions \ref{as:buffer} and \ref{as:dense} hold, and that the routing policy is either JFSQ or JFIQ. Then for a sufficiently large $N$, the following results hold:
\begin{enumerate}
\item[(i)] the expected number of jobs in servers of the first $K$ types divided by $N$ is bounded as 
\begin{equation}\label{eq:njobs-firstK}
\expect{\max\left(\sum_{m=1}^K C_m(\bar{\mathbf{Q}}) - (C^* +\epsilon),0\right)} \leq \frac{52\tau_{1K}b^2}{\epsilon N};
\end{equation}
\item[(ii)] 
if $K < M$, the expected number of jobs in the system divided by $N$ is bounded as 
\begin{equation}\label{eq:bound-total}
\expect{\sum_{m=1}^M C_m(\bar{\mathbf{Q}})} \leq C^*+\left(1+\frac{\tau_{KM}}{2}\right)\epsilon+2\sqrt{\frac{5\tau_{1M}b\ln N}{ N}}+60b^2\sqrt{\frac{26\tau_{1K}\tau_{1M}}{\hat{\beta}\epsilon N}};
\end{equation}
\item[(iii)]
the probability $p_{\set{B}}$ that an arriving job is blocked is bounded as 
\begin{equation}\label{eq:finite-block}
p_{\set{B}} \leq \frac{\tilde{d_2}}{\lambda} + \frac{52\tau_{1K}b^2}{\epsilon N}.
\end{equation}
\end{enumerate}
\end{theorem}

\subsection{Asymptotic Optimality}

Theorem~\ref{thm:finite-bound} may be difficult to interpret since there are several parameters involved in the results. So let us interpret the result for an important special case which is perhaps the one that is practically most relevant. Suppose that the normalized arrival rate $\lambda$, the proportions of different types of servers $\{\alpha_m\}$, and $\epsilon$ are fixed. In most practical systems, the number of jobs that can wait at a server is small, so let us suppose that $b$ is a fixed constant satisfying Assumption~\ref{as:dense}. Then, from (\ref{eq:bound-total}), it is clear that the normalized expected number of jobs in the system is asymptotically equal to $C^* + O(\epsilon)$ in the many-server limit. The blocking probability goes to zero provided $\tilde{d}_2=o(1)$ and the rate at which it goes to zero depends on rate at which $\tilde{d}_2$ decreases with $N.$ From Proposition~\ref{thm:lower-bound}, the lower bound on the normalized number of jobs in an infinite buffer system is $C^*.$ This suggests that JFSQ and JFIQ are near-optimal from the perspective of mean response time if the graph is reasonably well connected; we make this argument more general (by allowing many parameters to scale) and precise next.

To study the limit as $N$ approaches infinity, we let $\{G_N=(\set{L}_N,\set{R}_N,E_N), N \geq 1\}$ be a sequence of bipartite graphs such that $|\set{R}_N| = N$ and the buffer size of each server is given by $b_N$. Here, the number of servers, $N$, is allowed to scale, but the server-type distribution $(\alpha_1,\cdots,\alpha_M)$, and the service rate of each type of servers, $(\mu_1,\cdots,\mu_M), \mu_1 > \cdots > \mu_M$, are fixed. Further, the total arrival rates at ports in $\set{L}_N$, $\lambda_{\Sigma}$, is assumed to be equal to $N\sum_{m=1}^K \mu_m\alpha_m (1-\beta_N)$ for all $G_N$. As before, we can define a sequence of parameters $\{\epsilon_N, N \geq 1\}$ that quantify the approximation error where $\epsilon_N \in (0,\frac{\hat{\beta}_N}{4}]$, and $\hat{\beta}_N = \beta_N\sum_{m=1}^K \alpha_m.$ Now we can discuss the asymptotic performance of a routing policy as $N \to \infty$.

Proposition \ref{thm:lower-bound} provides a lower bound on the expected service time of a job in the system with infinite buffers. we thus have the following definition of an (asymptotically) optimal routing policy in the bipartite load balancing system.
\begin{definition}[Optimality in the Mean Response Time Sense] \label{def:optimality}
A stable routing policy is asymptotically optimal in the response time if the mean response time of jobs converges to $\frac{C^*}{\lambda}$ and the blocking probability goes to zero when $N \to \infty$. 
\end{definition}
We can see that optimality in the mean response time is a stronger metric than the common zero-waiting property discussed in the literature \cite{stolyar_2015,gamarnik_2018, LiuYing2020}. With this optimality, not only an arriving job has asymptotically zero waiting time, but it also has the minimum possible service time.

Then Theorem \ref{thm:finite-bound} immediately implies that both JFSQ and JFIQ are asymptotically optimal if the load of the system is moderate and the graph $G_N$ is suitably well connected. 
\begin{corollary}\label{cor:asymptotic}
Suppose that $\epsilon_N$ is both $o(1)$ and $\omega(ln(N)N^{-0.5})$, and that both Assumptions \ref{as:buffer} and \ref{as:dense} hold for $G_N$ when $N$ is sufficiently large. Then as $N \to \infty$, both JFSQ and JFIQ are asymptotically optimal, and the expected queueing delay converges to zero for both policies.
\end{corollary}
Due to the relationship between $\beta_N$ and $\epsilon_N,$ it is not difficult to see that asymptotic optimality holds for arrival rates upto the sub-Halfin-Whitt regime. We refer the reader to the appendix for a proof of Corollary \ref{cor:asymptotic}.

\subsection{Random Graph Models}
We now discuss when a bipartite graph can satisfy Assumption \ref{as:dense} in random graph models. Suppose the set of ports $\set{L}$ and the set of servers $\set{R}$ are fixed, but connections between them, i.e., the graph $G$, is not determined. This section considers a random graph $G$ where port $i$ connects with server $j$ with probability $z_{ij}$. We devise an explicit construction of $z_{ij} $ and show that such a random graph can satisfy Assumption \ref{as:dense} with a high probability. Our result first provides the construction of $z_{ij}$ when ports can have different arrival rates. Later, by restricting the scope to homogeneous arrival rates among ports, we give a better construction where the graph $G$ can have fewer edges. We are now ready to state our results.
\begin{theorem}\label{thm:random-graph}
Let $H_j = \frac{2\ln{2}(N+L)/N}{p_j}$ for $j \in \{1,2\}$. Consider the following construction of the graph $G$. For each port $\ell \in \set{L}$, 
\begin{itemize}
\item if $\lambda_{\ell} \geq N\frac{\tilde{d}_1}{H_1}$, this port connects with all servers of types less than $K$;
\item if $\lambda_{\ell} \geq N\frac{\tilde{d}_2}{H_2}$, this port connects with all servers of types equal to $K$;
\item otherwise, for each server $r \in \set{R}$, if $r \in \set{R}_{K-1}$, then $\ell$ connects with $r$ with probability $\frac{\lambda_{\ell} H_1}{N\tilde{d}_1}$. And if $r \in \set{R}_K \setminus \set{R}_{K-1}$, then $\ell$ connects with $r$ with probability $\frac{\lambda_{\ell}H_2}{N\tilde{d}_2}$.
\end{itemize}
Then $G$ satisfies Assumption \ref{as:dense} with probability at least $1 - 2^{-(N+L-1)}.$ The expected total number of edges used in $G_N$ scales as $O(\frac{(N+L)b^5}{\epsilon^2})$. 
\end{theorem}

Next, we discuss the special case of homogeneous arrival rates.

\begin{theorem}\label{thm:random-graph-uniform}
Suppose that all ports have the same arrival rates, that is, $\lambda_{\ell}\equiv \bar{\lambda}$ for all $\ell \in \set{L}$. Then following the same construction of graph $G$ in Theorem \ref{thm:random-graph} but with $H_j = 6\left(-\ln{p_j}+\frac{\tilde{d_j}}{p_j\bar{\lambda}}\ln\frac{2\mu_1}{\tilde{d_j}}\right)$ for $j \in \{1,2\}$, it holds that $G$ satisfies Assumption \ref{as:dense} with probability at least $1 - 2\binom{N}{Np_1}^{-1}$. The total number of edges in $G_N$ scales as $O\left(\frac{(N+L)b^3}{\epsilon}\ln\frac{b}{\epsilon}\right)$.
\end{theorem}
\begin{remark}
Th previous two theorems indicate that to achieve asymptotically optimal mean response time and asymptotic zero waiting probability, the average number of  connections of each port is only $O(\frac{1}{\epsilon^2})$ for heterogeneous arrival rates, and $O(\frac{1}{\epsilon}\ln\frac{1}{\epsilon})$ for homogeneous arrival rates, given that $L = \Omega(N), b = O(1)$. When $1/(1-\lambda)=O(1),$ we only require $\epsilon=o(1)$. Then the average number of edges connected to each port becomes $\omega(1)$. Therefore, for achieving very small loss probability and near-optimal response times, the number of edges in a random graph need to be only sparse compared to a fully connected graph. 
\end{remark}

\section{Proof of the Upper Bound and Optimality Results}\label{sec:upper-bound}
In this section, we provide the proofs of Theorem \ref{thm:finite-bound}. These results respectively bound the mean number of jobs in a finite-size system and show the asymptotic optimality for JFSQ and JFIQ in the many-server limit and the sub Halfin-Whitt regime. 

\subsection{Proof Sketch}
Ahead of the complete proof, we first provide a sketch of the proof reflecting intuitions behind it. Recall that the goal is to bound the mean number of jobs in the system divided by $N$, given by $\expect{\sum_{m=1}^M C_m(\bar{\bold{Q}})}$. Here by definition, $C_m(\bar{\bold{Q}}) = \sum_{j = 1}^b s_{m,j}(\bar{\bold{Q}})$. Our proof starts with the following observation about the mean-field limit for JFSQ and JFIQ in the heterogeneous system. 
\subsubsection{Mean-Field Limit}
Ideally, if the load $\lambda$ is a constant, then as $N \to \infty$, it holds that 
\begin{equation}\label{eq:true-mean-field}
s_{m,1}(\bar{\bold{Q}}) \approx \left\{
\begin{aligned}
\alpha_m, & ~m < K \\
C_K^*, & ~m = K \\
0, & ~m > K
\end{aligned}
\right.\mspace{50mu}\text{and}\mspace{50mu}
s_{m,j}(\bar{Q}) \approx 0,~\forall {m=1\dots M, j=2 \dots b}.
\end{equation}
Roughly speaking, this limit tells us that all the first $K-1$ types of servers are busy, some servers of type $K$ are busy, and all the servers with types greater than $K$ are idle. 

The intuition behind (\ref{eq:true-mean-field}) is as follows. Since there are infinite servers, a certain fraction of them must be idle. Then by the definition of JFIQ and JFSQ, all arrivals of jobs are routed to idle servers, at least in a fluid model. Therefore, the scaled number of waiting jobs (i.e., not in service), $\sum_{m=1}^M\sum_{j=2}^b S_{m,j}(\bold{Q})$ must converge to zero. For $S_{1,1}(\bold{Q}),\cdots,S_{M,1}(\bold{Q})$, JFIQ and JFSQ always route jobs to fastest idle servers. Therefore, it must be the case that $s_{m,1}(\bold{Q})$ are filled from $1$ to $M$ until $\sum_{m=1}^M \mu_m s_{m,1}(\bar{\bold{Q}}) = \lambda$. That is to say, the total departure rate is equal to the total arrival rate. Therefore, we can `guess' that the mean-field limit has the form (\ref{eq:true-mean-field}). 

Based on this limit, the scaled mean number of jobs can be decomposed as 
\begin{equation}
\expect{\sum_{m=1}^M C_m(\bar{\bold{Q}})} = \expect{\sum_{m=1}^K C_m(\bar{\bold{Q}})} + \expect{\sum_{m=K + 1}^M C_m(\bar{\bold{Q}})}.
\end{equation}
\subsubsection{Lyapunov Drift Arguments}

The drift argument starts by considering a Lyapunov function $g$ and setting its drift in steady-state equal to zero. Since we are considering continuous-time Markov chains, this is equivalent to saying that
$\expect{Gg(\bar{\bold{Q}})}=0$ where $G$ is the generator of the Markov chain (defined explicitly later). Initially, let us focus on the total queue length in the first $K$ types of servers (scaled by $N$) and thus, choose the Lyapunov function to be a function of the scaled total number of jobs in these servers and their queues, which we will call $x$. By an abuse of notation, we will rewrite the drift as
$\expect{Gg(x)}=0.$
However, this drift may be hard to analyze. Instead, suppose that the system was a simple deterministic fluid model of the form $\dot{x}=-\Delta$ for an appropriately $\Delta>0.$ The motivation for considering this fluid model is that, in the large-system limit, our system behaves like a single-server queue with simple fluid dynamics. If this fluid limit were the true system, then the drift of $g$ becomes simply $-g'(x)\Delta.$ We add and subtract this drift from the drift of the stochastic system to obtain
$\expect{Gg(x)-g'(x)\Delta+g'(x)\Delta}=0,$
which can be rewritten as
$$\expect{g'(x)\Delta}=\expect{Gg(x)-(-g'(x)\Delta)}.$$
We are interested in getting a bound on the steady-state expectation of $h(x)=(x-C^*+\epsilon)^+$ where $\epsilon$ controls the approximation error. Therefore, we choose $g$ such that $g'(x)\Delta=h(x)$ (this equality is sometimes called Stein's equation). Thus, the drift equation becomes
$$\expect{h(x)}=\expect{Gg(x)-(-g'(x)\Delta)}.$$
Now, it is easy to see that we can bound $\expect{h(x)}$ if we can show that the drift of the Markov process $\expect{G(g(x))}$ is approximately equal to $-g'(x)\Delta.$ The rest of the proof involves studying $\expect{Gg(x)-(-g'(x)\Delta)}$ by choosing $\Delta=\mu_1\delta$ where $\delta>0.$

In Lemma \ref{lemma:firstK-key-term}, we show that this expression is approximately equal to 
\begin{equation}\label{eq:approximate-key-term}
 \frac{1}{\mu_1\delta} \expect{\mathbbm{1}\left\{\sum_{m=1}^K C_m(\bar{\bold{Q}}) \geq C^*+\epsilon + \frac{1}{N}\right\} h\left(\sum_{m=1}^K C_m(\bar{\bold{Q}})\right)(\lambda+\mu_1\delta - W(\bar{\bold{Q}}))}.
\end{equation}
We want to upper bound this expression by a quantity which is small when $N$ is large.
Note that $\sum_{m=1}^K C_m(\bar{\bold{Q}})$ is the total scaled queue length in the first $K$ types of servers and $W(\bar{\bold{Q}}) = \sum_{m=1}^K \mu_m s_{m,1}(\bar{\bold{Q}})$ can be interpreted as the departure rate from these servers. Thus, the above expression can be upper bounded by a small quantity if the following holds: whenever the total queue length is large, the departure rate exceeds the arrival rate with high probability.

To establish this fact, the mean-field limit (\ref{eq:true-mean-field}) motivates us to show that $s_{m,1}(\bar{\bold{Q}}) \approx \alpha_m$ for $m < K$ and $s_{K,1}(\bar{\bold{Q}}) \approx C_K^*$. To be concrete, we show a two-stage state space collapse result through the following two Lyapunov functions (omitting extra technical terms):

\begin{align}
\tilde{V}_1(\bold{q}) &= \min\left(\sum_{m=1}^{K-1}\sum_{j=2}^b s_{m,j}(\bold{q}) + C_K(\bold{q}), \sum_{m=1}^{K-1} \alpha_m - \sum_{m=1}^{K-1} s_{m,1}(\bold{q})\right) \label{eq:tildeV1}\\
\tilde{V}_2(\bold{q}) &= \min\left(\sum_{m=1}^{K}\sum_{j=2}^b s_{m,j}(\bold{q}), \sum_{m=1}^{K-1} C_m^* + \tau_{1K}\delta - \sum_{m=1}^{K} s_{m,1}(\bold{q})\right) \label{eq:tildeV2}.
\end{align}

The well-connectedness condition in Assumption \ref{as:dense} and the routing policy (JFSQ and JFIQ) ensure that both of them have negative drifts when they are sufficiently large (Lemma \ref{lemma:lyapunov-firstK-1} and Lemma \ref{lemma:lyapunov-firstK}). We now provide some intuition to explain how the well-connectedness condition plays a role in establishing the negative drift of these Lyapunov functions. We consider $\tilde{V}_1,$ the explanation for the other Lyapunov function is similar. If $\tilde{V}_1$ is large, it implies that both terms inside the min in (\ref{eq:tildeV1}) are large. In particular, by focusing on the second term, we note that a large $\tilde{V}_1$ implies that the (scaled) number of used servers $\sum_{m=1}^{K} s_{m,1}(\bold{q})$ is small. Equivalently, the number of idle servers is large. The well-connected condition simply states that the arrival rates to large subsets of servers is large. Thus, if $\tilde{V}_1$ is large, the number of empty servers is large which implies they have a large arrival rate, which in turn implies that the number of empty servers quickly decreases.
The negative drift of $\tilde{V}_1$ and $\tilde{V}_2$ can be used to establish geometric tail bounds (Lemma \ref{lemma:geometric-tail}) using standard drift arguments to show that they are small with high probability. 

Observe that when $\sum_{m=1}^K C_m(\bold{q}) > C^*+\epsilon$, these two Lyapunov functions are all equal to the second term on their right hand side. 
Then in this case, $\sum_{m=1}^{K-1} s_{m,1}(\bold{q}) \approx \sum_{m=1}^{K-1} \alpha_m$, and $\sum_{m=1}^{K} s_{m,1}(\bold{q}) \approx \sum_{m=1}^{K} C_m^* + \tau_{1K}\delta$. It then implies $s_{K,1}(\bold{q}) \approx C_K^*+\tau_{1K}\delta$. Now that $\sum_{m=1}^K \mu_m C_m^* = \lambda$, it holds $W(\bold{q}) \approx \lambda + \mu_1\delta$ with high probability. We thus prove that (\ref{eq:approximate-key-term}) should be small, and it leads to a bound on the scaled mean number of jobs in the first $K$ types of servers.
	
Now for the remaining types of servers, the mean-field limit (\ref{eq:true-mean-field}) indicates that almost all of them are idle. We thus try to bound this third Lyapunov function, $\sum_{m=K+1}^M C_m(\bar{\bold{Q}})$. From the mean-field limit, we know that $\sum_{m=1}^K s_{m,1}(\bold{Q}) \approx C^*$. Therefore, approximately $N\left(\sum_{m=1}^K \alpha_m - C^*\right)$ servers of the first $K$ types are idle. Therefore, Assumption \ref{as:dense} ensures that very few jobs are routed to the remaining types of servers under JFSQ and JFIQ. By utilizing a conditional geometric tail bound (Lemma \ref{lemma:geometric-tail}), we manage to show that $\sum_{m=K+1}^M C_m(\bar{\bold{Q}})$ is small with high probability, and finally obtain a bound on its mean. 

For the complete proof of Theorem \ref{thm:finite-bound}, since our theorem consists of three parts, we prove each of them in order, and combine them together at the end of this section. 

\subsection{Bound for the First $K$ Types of Servers}
The first result, which bounds the number of jobs in the first $K$ types of servers, is the most important part in the theorem, which is restated as follows.
\begin{lemma}\label{lemma: njobs-firstK}
Under Assumption \ref{as:buffer} and Assumption \ref{as:dense}, the expected number of jobs in servers of the first $K$ types divided by $N$ is bounded as 
\begin{equation}
\expect{\max\left(\sum_{m=1}^K C_m(\bar{\mathbf{Q}}) - (C^* +\epsilon),0\right)} \leq \frac{52\tau_{1K}b^2}{\epsilon N} \tag{\ref{eq:njobs-firstK}}
\end{equation}
if the routing policy is either JFSQ or JFIQ.
\end{lemma}
\begin{proof}
Throughout this proof, we assume all assumptions in Lemma \ref{lemma: njobs-firstK} are satisfied. Recall that the metric of interest is $\expect{\max\left(\sum_{m=1}^K C_m(\bar{\mathbf{Q}})-(C^*+\epsilon),0)\right)}$, where $C^* = \sum_{m=1}^K C_m^*.$ To simplify the notation, let $\eta = C^*+\epsilon$, and denote $h(x)=\max(x-\eta,0)$. Our goal is thus to bound $\expect{h(\sum_{m=1}^K C_m(\bar{\mathbf{Q}}))}$. The proof is motivated by the framework introduced in \cite{LiuYing2020}, and can be divided mainly into three parts, generator coupling, gradient bounds and state-space collapse.
\paragraph{\textbf{Generator Coupling}}
We couple our system with a fluid model that is simple, but can well approximate the evolution of $h(\sum_{m=1}^K C_m(\bar{\mathbf{Q}}))$. In particular, consider a fluid model $\dot{x} = -\mu_1\delta$ where $\delta = \frac{\mu_K}{6\mu_1b^2}\epsilon$. Let $g(x)$ be the solution to the following Stein's equation of the fluid model,
\begin{equation}
    \mu_1\delta g'(x) = h(x).
\end{equation}
The solution is unique, and is given by 
\begin{equation}\label{eq:g-gradient}
g(x) = \frac{\max(x-\eta,0)^2}{2\mu_1\delta},~ g'(x)=\frac{\max(x-\eta,0)}{\mu\delta},~g''(x)=\left\{
\begin{aligned}
0,~x<\eta \\
\frac{1}{\mu_1\delta},~x\geq \eta.
\end{aligned}
\right.
\end{equation}
The next step is to couple our system with the fluid model through this stein's equation. 

To do so, recall that the system is a CTMC defined on queue lengths of servers, $\mathbf{Q}(t)$. let $G$ be the generator of our system such that for a queue state $\mathbf{q}$, and any function $V$ defined on the state space, 
\begin{equation}
GV(\mathbf{q}) = \sum_{\mathbf{q'}} r_{\mathbf{q},\mathbf{q'}} \left(V(\mathbf{q'})-V(\mathbf{q})\right)
\end{equation}
where $r_{\mathbf{q},\mathbf{q'}}$ is the transition rate from state $\mathbf{q}$ to state $\mathbf{q'}$. It is clear that $Gg(\mathbf{q})$ serves as an analog of the drift of function $g$ at state $\mathbf{q}$ in a discrete-time Markov chain as in \cite{EryilmazSrikant2012}. To couple our system with the fluid model, we first need the following property, a key insight from \cite{EryilmazSrikant2012} and \cite{LiuYing2020}.
\begin{lemma}\label{lemma:zero-drift}
The expectation $\expect{Gg(\sum_{m=1}^K C_m(\bar{\mathbf{Q}}))}$ is equal to $0$.
\end{lemma}

Then the two systems can be coupled by seeing that 
\begin{align}
\expect{h\left(\sum_{m=1}^K C_m(\bar{\mathbf{Q}})\right)} &= \expect{g'\left(\sum_{m=1}^K C_m(\bar{\mathbf{Q}})\right)(\mu_1\delta)} \\
&=\expect{Gg\left(\sum_{m=1}^K C_m(\bar{\mathbf{Q}})\right) - g'\left(\sum_{m=1}^K C_m(\bar{\mathbf{Q}})\right)(-\mu_1\delta) }\label{eq:geneartor-diff}.
\end{align}
As a result, to bound $\expect{h\left(\sum_{m=1}^K C_m(\bar{\mathbf{Q}})\right)}$, it is equivalent to bound (\ref{eq:geneartor-diff}).
\paragraph{\textbf{Gradient Bounds.}}
We now utilizing the explicit form of $g(x)$ in (\ref{eq:g-gradient}) to bound (\ref{eq:geneartor-diff}). First by definition, it holds that for a state $\mathbf{q}$,
\begin{align}
Gg\left(\sum_{m=1}^K C_m(\mathbf{q})\right) &= \sum_{\mathbf{q'}} r_{\mathbf{q},\mathbf{q'}} \left(g\left(\sum_{m=1}^K C_m(\mathbf{q'})\right)-g\left(\sum_{m=1}^K C_m(\mathbf{q})\right)\right) \notag \\
&= \lambda_{\Sigma}(1-P_k(\mathbf{q}))\left(g\left(\sum_{m=1}^K C_m(\mathbf{q})+\frac{1}{N}\right)-g\left(\sum_{m=1}^K C_m(\mathbf{q})\right)\right) \mspace{40mu}(\text{Arrival transitions}) \\ \label{eq:generator-arrival}
&\mspace{20mu}+NW(\mathbf{q})\left(g\left(\sum_{m=1}^K C_m(\mathbf{q})-\frac{1}{N}\right)-g\left(\sum_{m=1}^K C_m(\mathbf{q})\right)\right) \mspace{40mu}(\text{Departure transitions})
\end{align}
where $P_k(\mathbf{q})$ is the probability that an arrival of jobs is not routed to a server of type no greater than $K$, and $W(\mathbf{q})=\sum_{m=1}^K \mu_m s_{m,1}(\mathbf{q})$. Then by (\ref{eq:geneartor-diff}), we can get 
\begin{align}
\expect{h\left(\sum_{m=1}^K C_m(\bar{\mathbf{Q}})\right)}  &\leq \mathbb{E}\left[g'\left(\sum_{m=1}^K C_m(\bar{\mathbf{Q}})\right)(\mu_1\delta)\right. \label{eq:initial-value}\\
&\left.\mspace{20mu}+\lambda_{\Sigma}\left(g\left(\sum_{m=1}^K C_m(\bar{\mathbf{Q}})+\frac{1}{N}\right)-g\left(\sum_{m=1}^K C_m(\bar{\mathbf{Q}})\right)\right) \right. \label{eq:arrival-term}\\
&\left.\mspace{20mu} +NW(\bar{\mathbf{Q}})\left(g\left(\sum_{m=1}^K C_m(\bar{\mathbf{Q}})-\frac{1}{N}\right)-g\left(\sum_{m=1}^K C_m(\bar{\mathbf{Q}})\right)\right)\right] \label{eq:departure-term}
\end{align}
where we omit the term $P_k(\bar{\bold{Q}})$ from (\ref{eq:generator-arrival}) since $g(x)$ is an increasing function by (\ref{eq:g-gradient}). 
Now to simplify the equation, we can do Taylor's expansion on (\ref{eq:arrival-term}) and (\ref{eq:departure-term}), and apply gradient bounds of $g(x)$. The result is summarized as follows whose proof is provided in the appendix.
\begin{lemma}\label{lemma:firstK-key-term}
It holds that
\begin{equation}\label{eq:refined-key-term}
\expect{h\left(\sum_{m=1}^K C_m(\bar{\mathbf{Q}})\right)} \leq \expect{\mathbbm{1}\left\{\sum_{m=1}^K C_m(\bar{\bold{Q}}) \geq \eta + \frac{1}{N}\right\} g'\left(\sum_{m=1}^K C_m(\bar{\bold{Q}})\right)(\mu_1\delta+\lambda-W(\bar{\bold{Q}}))} + \frac{38b^2\tau_{1K}}{\epsilon N}.
\end{equation}
\end{lemma} 

The remaining step is to bound the first term on the right hand side in (\ref{eq:refined-key-term}), which is the main part of this proof. The key insight is that as long as $W(\bold{q}) \geq \lambda +\mu_1\delta$, it holds that the contribution of $\bold{q}$ to the first term would be at most zero. Furthermore, this property only needs to hold when $\sum_{m=1}^K C_m(\bold{q}) \geq \eta + \frac{1}{N}$ due to the indicator function. To justify this result, we establish two state space collapse results as follows.

\paragraph{\textbf{State Space Collapse.}}
Recall that $\sum_{m=1}^K C_m(\bold{q})$ is the number of jobs in servers of the first $K$ types divided by $N$. The intuition is to show that when this number is large, it holds that with high probability, 
\begin{equation}\label{eq:firstK-intuition}
s_{1,1}(\bold{q}) = C_1^*,\cdots, s_{K-1,1}(\bold{q})=C_{K-1}^*, s_{K,1} > C_K^*.
\end{equation}
That is to say, almost all servers of the first $K-1$ types are busy. And enough type-$K$ servers are busy such that their total departure rates (or works produced by these servers) are sufficient for the total arrival rate $\lambda_{\Sigma}$.

The following lemma indirectly shows that unless $\sum_{m=1}^K C_m(\bold{q})$ is small, $\sum_{m=1}^K s_{m,1}(\bold{q}) \approx \sum_{m=1}^{K-1} \alpha_m$. In particular, it designs a Lyapunov function closely related to the above property. Due to space limitations, the proof is deferred to the appendix.
\begin{lemma} \label{lemma:lyapunov-firstK-1}
Consider the following Lyapunov function
\begin{equation}\label{eq:lyapunov-firstK-1}
V_1(\bold{q}) = \min\left(\sum_{j=1}^b s_{K,j}(\bold{q})+\sum_{m=1}^{K-1}\sum_{j=2}^b s_{m,j}(\bold{q}), \sum_{m=1}^{K-1} C_m^*-\sum_{m=1}^{K-1} s_{m,1}(\bold{q})\right).
\end{equation}
It holds that if $V_1(\bold{q}) \geq B_1 \coloneqq \tau_{1K}\delta$, then $GV_1(\bold{q}) \leq \frac{-\mu_1\delta}{2b}$.
\end{lemma}

In addition to Lemma \ref{lemma:lyapunov-firstK-1} that focuses on the first $K-1$ types of servers, the following lemma provides another Lyapunov function. This function is later used together with Lemma \ref{lemma:lyapunov-firstK-1} to show that if $\sum_{m=1}^K C_m(\bold{q})$ is large, then a certain number of type $K$ servers are busy. It then complements the goal in (\ref{eq:firstK-intuition}). The proof of this lemma is similar to that of Lemma \ref{lemma:lyapunov-firstK-1}, and is provided in the appendix.
\begin{lemma}\label{lemma:lyapunov-firstK}
Consider the following Lyapunov function
\begin{equation}
V_2(\bold{q}) = \min\left(\sum_{m=1}^K \sum_{j=2}^b s_{m,j}(\bold{q}), \sum_{m=1}^K C_m^* + B_2 + 3\tau_{1K}\bar{\delta}-\sum_{m=1}^K s_{m,1}(\bold{q})\right) \label{eq:lyapunov-firstK}
\end{equation}
where $\bar{\delta} \coloneqq \tau_{1K}\delta$, and $B_2 \coloneqq \frac{1}{2}\epsilon + \bar{\delta}.$ It holds that if $V_2(\bold{q}) \geq B_2$, then $GV_2(\bold{q}) \leq -\frac{\mu_1\delta}{b}.$
\end{lemma}

To apply the above two lemmas, we need the following geometric tail bound from \cite{WengWang20}, which originates in  \cite{BertsimasGarmarnik01,WangMaguluri18}. This lemma translates the fact that a Lyapunov function has a negative drift to the property that the function is within a certain region with high probability. 

\begin{lemma}\label{lemma:geometric-tail}
	Consider a continuous time Markov chain $\left\{\bold{S}(t):t\ge 0\right\}$ on a finite state space $\mathcal{S}$. Assume that it has a unique stationary distribution.  For a Lyapunov function $V: \mathcal{S} \rightarrow [0,+\infty)$, define
	$
	GV(\bold{s}) = \sum_{\bold{s}' \in \mathcal{S}} r_{\bold{s},\bold{s}'}(V(\bold{s}') - V(\bold{s}))
	$
	where $r_{\bold{s},\bold{s}'}$ is the transition rate from state $\bold{s}$ to $\bold{s}'$.  
	
	Suppose that
	\[
	\nu_{\mathrm{max}} \coloneqq \sup_{\bold{s},\bold{s}' \in \mathcal{S}\colon r_{\bold{s},\bold{s}'} > 0} |V(\bold{s}) - V(\bold{s}')| < \infty;~~
	f_{\mathrm{max}} \coloneqq \max\left\{0,\sup_{\bold{s} \in \mathcal{S}} \sum_{\bold{s}':V(\bold{s}') > V(\bold{s})} r_{\bold{s} ,\bold{s}'}\left(V(\bold{s}') - V(\bold{s})\right)\right\} < \infty.
	\]
	Given a set $\mathcal{E}$. If for some $B > 0,\gamma > 0, \xi \geq 0$, it holds: 1) $G V(\bold{s}) \leq -\gamma$ when $V(\bold{s}) \geq B$ and $\bold{s} \in \mathcal{E}$; 2) $GV(\bold{s}) \leq \xi$ when $V(\bold{s}) \geq B$ and $\bold{s} \not \in \mathcal{E}$,

	then for all positive integer $j$, if $\bar{\bold{S}}$ is the steady-state random variable, it holds 
	\begin{equation}
	\Pr\left\{V(\bar{\bold{S}}) \geq B + 2 \nu_{\mathrm{max}}j\right\} \leq \left(\frac{f_{\mathrm{max}}}{f_{\mathrm{max}} + \gamma}\right)^j + \left(\frac{\xi}{\gamma} + 1\right)\Pr\left\{s \not \in \mathcal{E}\right\}.
	\end{equation}
	
\end{lemma}

Based on Lemma \ref{lemma:geometric-tail}, we can bound the probability that $V_1(\bold{q})$ or $V_2(\bold{q})$ is large in the following result.
\begin{lemma}\label{lemma:firstK-collapse}
Let $\chi = 96\tau_{1K}b^3\ln N$. With the same notation in Lemma \ref{lemma:lyapunov-firstK-1} and Lemma \ref{lemma:lyapunov-firstK}, it holds that 
\begin{equation} 
\Pr\left\{V_1(\bar{\bold{Q}}) \geq B_1+\frac{\chi}{\epsilon N}\right\} \leq N^{-2};
\Pr\left\{V_2(\bar{\bold{Q}}) \geq B_2+\frac{\chi}{\epsilon N}\right\} \leq N^{-2}.
\end{equation}
\end{lemma}
\begin{proof}
Note that under the notation in Lemma \ref{lemma:geometric-tail}, we have for both $V_1(\bold{q})$ and $V_2(\bold{q})$, $\nu_{\max} = \frac{1}{N}$, and $f_{\max} \leq \mu_1$. We first bound $\Pr\left\{V_1(\bold{q}) \geq B_1+\frac{\chi}{\epsilon N}\right\}$. Since by Lemma \ref{lemma:lyapunov-firstK-1}, when $V_1(\bold{q}) \geq B_1$, it holds $GV_1(\bold{q}) \leq \frac{-\mu_1\delta}{2b}$. Then by taking the set $\set{E}$ to be the empty set and taking $j_1 = \frac{8b}{\delta}\log N$, Lemma \ref{lemma:geometric-tail} shows that 
\begin{equation}
\Pr\left\{V_1(\bold{q}) \geq B_1+2\nu_{\max}j_1\right\} \leq \left(1+\frac{\delta}{2b}\right)^{-j_1} \leq \exp\left(-\frac{j_1\delta}{4b}\right)=N^{-2}
\end{equation}
where the last inequality comes from the fact that $\ln(1+x) \geq x/2$ for $x \in [0,1]$. We can easily verify that 
$
2\nu_{\max}j_1 = \frac{2}{N}\cdot \frac{48\mu_1b^3}{\mu_K\epsilon}=\frac{\chi}{\epsilon N}.
$
Similarly, take $j_2 = \frac{4b}{\delta}\log N$ for $V_2(\bold{q})$. Together with Lemma \ref{lemma:lyapunov-firstK}, Lemma \ref{lemma:geometric-tail} shows that 
\begin{equation}
\Pr\left\{V_2(\bold{q}) \geq B_2+2\nu_{\max}j_2\right\} \leq \left(1+\frac{\delta}{b}\right)^{-j_2} \leq \exp\left(-\frac{j_2\delta}{2b}\right)=N^{-2}.
\end{equation}
We complete the proof by noticing that 
$
2\nu_{\max}j_2 = \frac{2}{N} \cdot \frac{24\mu_1b^3}{\mu_K\epsilon} \leq \frac{\chi}{\epsilon N}.
$
\end{proof}

\paragraph{\textbf{Completing the Whole Proof}}
Finally, combining Lemma \ref{lemma:firstK-collapse} with Lemma  \ref{lemma:firstK-key-term} help us complete the proof. To see why, recall that it remains to bound 
\begin{equation}\label{eq:firstK-key-term}
\expect{\mathbbm{1}\left\{\sum_{m=1}^K C_m(\bar{\bold{Q}}) \geq \eta + \frac{1}{N}\right\} g'\left(\sum_{m=1}^K C_m(\bar{\bold{Q}})\right)(\lambda+\mu_1\delta-W(\bar{\bold{Q}}))}. 
\end{equation}
Let event $\set{D} = \{V_1(\bar{\bold{Q}}) \leq B_1 + \frac{\chi}{\epsilon N}\} \cap \{V_2(\bar{\bold{Q}}) \leq B_2 + \frac{\chi}{\epsilon N}\}$. It holds that 
\begin{align}
(\ref{eq:firstK-key-term}) &\leq \mathbb{E}\left[\mathbbm{1}\left\{\sum_{m=1}^K C_m(\bar{\bold{Q}}) \geq \eta + \frac{1}{N}\right\} g'\left(\sum_{m=1}^K C_m(\bar{\bold{Q}})\right)(\lambda+\mu_1\delta-W(\bar{\bold{Q}})) \middle| \set{D}\right]+g'(b)\mu_1(1+\delta)\Pr\{\bar{\set{D}}\} \notag \\
&\leq  \mathbb{E}\left[\mathbbm{1}\left\{\sum_{m=1}^K C_m(\bar{\bold{Q}}) \geq \eta + \frac{1}{N}\right\} g'\left(\sum_{m=1}^K C_m(\bar{\bold{Q}})\right)(\lambda+\mu_1\delta-W(\bar{\bold{Q}})) \middle| \set{D}\right] + \frac{2b}{\delta N^2}(1+\delta) \label{eq:firstK-rem1}
\end{align}
where the first inequality is by the law of total probability and the fact that $g'(x)$ is a positive increasing function, that $\sum_{m=1}^K C_m(\bold{q}) \leq b$ for all possible $\bold{q}$, and that $\lambda \leq \mu_1$, and the second inequality is by Lemma \ref{lemma:firstK-collapse} that shows $\Pr\{\bar{\set{D}}\} \leq \frac{2}{N^2}.$

Therefore, it is sufficient to bound the first term in (\ref{eq:firstK-rem1}). The following lemma shows that this term is indeed non-positive. 
\begin{lemma}\label{lemma:firstK-enoughWork}
For any $\bold{q}$ such that $V_1(\bold{q}) \leq B_1 + \frac{\chi}{\epsilon N}$ and $V_2(\bold{q}) \leq B_2 + \frac{\chi}{\epsilon N}$, it holds that 
\begin{equation}\label{eq:firstK-enoughWork}
\mathbbm{1}\left\{\sum_{m=1}^K C_m(\bold{q}) \geq \eta + \frac{1}{N}\right\} (\lambda+\mu_1\delta-W(\bold{q})) \leq 0.
\end{equation}
\end{lemma}
\begin{proof}
W.L.O.G., we can directly assume $\sum_{m=1}^K C_m(\bold{q}) \geq \eta + \frac{1}{N}$. Otherwise, (\ref{eq:firstK-enoughWork}) is already zero. Then the key step is to show $W(\bold{q}) = \sum_{m=1}^K \mu_m s_{m,1}(\bold{q}) \geq \lambda + \mu_1\delta$. By the definition of $V_1(\bold{q})$ in (\ref{eq:lyapunov-firstK}), since $\sum_{m=1}^K C_m(\bold{q}) \geq \eta + \frac{1}{N}$, it holds that $V_1(\bold{q}) = \sum_{m=1}^{K-1} C_m^* - \sum_{m=1}^{K-1} s_{m,1}(\bold{q})$. Furthermore, as $V_1(\bold{q}) \leq B_1 + \frac{\chi}{\epsilon N}$ and $C_m^* = \alpha_m$ for $m < K$, it satisfies 
\begin{equation}
\sum_{m=1}^{K-1} s_{i,1}(\bold{q}) \geq \sum_{m=1}^{K-1} \alpha_m - (B_1 + \frac{\chi}{\epsilon N}).
\end{equation}
Since $s_{m,1}(\bold{q}) \leq \alpha_m$ for all $m$, the total departure rate of servers of the first $K-1$ types is at least
\begin{equation}\label{eq:firstK-work-firstK-1}
\sum_{m=1}^{K-1} \mu_m s_{m,1}(\bold{q}) \geq \sum_{m=1}^{K-1} \mu_m \alpha_m - \mu_1\left(B_1+\frac{\chi}{\epsilon N}\right).
\end{equation}
Then for $s_{K,1}(\bold{q})$, recall the definition of $V_2(\bold{q})$ in (\ref{eq:lyapunov-firstK-1}). To show that $V_2(\bold{q})$ is equal to the second term in its definition, note that 
\[
B_2 + 3\tau_{1K}\bar{\delta} = \frac{1}{2}\epsilon + \tau_{1K}\delta+3\tau_{1K}^2\delta \leq \frac{1}{2}+\frac{2\tau_{1K}\epsilon}{3b^2} \leq \epsilon.
\]
Then since $\sum_{m=1}^K C_m(\bold{q}) \geq \sum_{m=1}^K C_m^*+\epsilon+\frac{1}{N}$, it holds $\sum_{m=1}^K C_m(\bold{q}) \geq \sum_{m=1}^K C_m^* + B_2 + 3\tau_{1K}\bar{\delta}$. Therefore, $V_2(\bold{q})$ is equal to $\sum_{m=1}^K C_m^* + B_2 + 3\tau_{1K}\bar{\delta}-\sum_{m=1}^K s_{m,1}(\bold{q})$, the second term in (\ref{eq:lyapunov-firstK-1}). By assumption, $V_2(\bold{q}) \leq B_2 + \frac{\chi}{\epsilon N}$. As a result, 
\begin{equation}
\sum_{m=1}^K s_{m,1}(\bold{q}) \geq \sum_{m=1}^K C_m^* + 3\tau_{1K}\bar{\delta} - \frac{\chi}{\epsilon N},
\end{equation}
and 
\begin{equation}\label{eq:firstK-work-firstK}
s_{K,1}(\bold{q}) \geq C_K^* + 3\tau_{1K}\bar{\delta} - \frac{\chi}{\epsilon N}
\end{equation}
because $s_{m,1}(\bold{q}) \leq \alpha_m = C_m^*$ for $m < K$. From (\ref{eq:firstK-work-firstK-1}) and (\ref{eq:firstK-work-firstK}), it holds
\begin{align}
W(\bold{q}) = \sum_{m=1}^{K-1} \mu_ms_{m,1}(\bold{q})+\mu_Ks_{K,1}(\bold{q}) 
&\geq \sum_{m=1}^{K-1}\mu_m\alpha_m+\mu_K C_K^* + 3\mu_K\tau_{1K}\bar{\delta}-\mu_1B_1 - 2\frac{\mu_1\chi}{\epsilon N} \\
&\geq \lambda + 2\frac{\mu_1^2}{\mu_K}\delta-\frac{192\mu_1^2b^3}{\mu_K \epsilon N}\ln(N) \geq \lambda + \mu_1\delta \label{eq:firstK-work-overflow}
\end{align}
where the last inequality is because $\mu_1 > \mu_K$, and $\frac{\mu_1^2}{\mu_K}\delta \geq \frac{192\mu_1^2\ln(N)}{\mu_K\epsilon N}b^3$ by Assumption \ref{as:buffer}. The inequality (\ref{eq:firstK-work-overflow}) immediately implies the desired result. 
\end{proof}
To conclude the proof of Lemma \ref{lemma: njobs-firstK}, by Lemma \ref{lemma:firstK-key-term}, the bound in (\ref{eq:firstK-rem1}) and Lemma \ref{lemma:firstK-enoughWork}, it holds
\begin{equation}
\expect{h\left(\sum_{m=1}^K C_m(\bar{\mathbf{Q}})\right)} \leq \frac{2b}{\delta N^2}(1+\delta) + \frac{38b^2\tau_{1K}}{\epsilon N} \leq \frac{12b^3\tau_{1K}}{\epsilon N^2} + \frac{2b}{N^2} + \frac{38b^2\tau_{1K}}{\epsilon N} \leq \frac{52b^2\tau_{1K}}{\epsilon N}.
\end{equation}
\end{proof}

\subsection{Bound for the Remaining Servers}
Since Lemma \ref{lemma: njobs-firstK} only bounds the mean number of jobs in servers of the first $K$ types, we need the following result for the remaining servers in the system. This result shows that very few jobs will be served by servers of the last $M-K$ types of jobs. Note that if $K = M$, then Lemma \ref{lemma: njobs-firstK} already bounds the mean number of jobs in the system. 

\begin{lemma}\label{lemma:njobs-lastM-K}
Suppose $K< M$. Under Assumption \ref{as:buffer} and Assumption \ref{as:dense}, if $N$ is sufficiently large, the expected number of jobs in servers of the last $M-K$ types divided by $N$ is bounded as 
\begin{equation}
\expect{\sum_{m=K+1}^M C_m(\bar{\mathbf{Q}})} \leq \frac{\tilde{d_2}b}{\mu_M}+2\sqrt{\frac{5\tau_{1M}b\ln N}{ N}}+8b^2\sqrt{\frac{26\tau_{1K}\tau_{1M}}{\hat{\beta}\epsilon N}}.
\end{equation}
if the routing policy is either JFSQ or JFIQ.
\end{lemma}
\begin{proof}
To prove this result, let us consider the Lyapunov  function $V_3(\bold{q})=\sum_{m=K+1}^M C_m(\bold{q})$. Then by showing that this function has a negative drift when outside of a region, we can obtain a bound on its expectation. To do so, define $B_3$ as 
\begin{equation}\label{eq:def-B3}
B_3=\frac{1}{\mu_M}\left(\tilde{d}_2b+\sqrt{\mu_1\mu_M\left(\frac{5b\ln(N)}{N}+\frac{416\tau_{1K} b^4}{\hat{\beta}\epsilon N}\right)}\right).
\end{equation}
Let $\set{E}_K=\{\bold{q}\colon \sum_{m=1}^K C_m(\bold{q}) \leq C^*+\frac{\hat{\beta}}{2}\}$. It holds that $\bar{\bold{Q}}$ lies in $\set{E}_K$ with high probability by the following lemma whose proof is in the appendix.
\begin{lemma} \label{lemma:prob-serverK-explode}
For any $\Delta \geq \frac{\hat{\beta}}{2}$, it holds $\Pr\{\sum_{m=1}^K C_m(\bar{\bold{Q}}) > C^*+\Delta\} \leq \frac{104\tau_{1K}b^2}{\Delta \epsilon N}.$
\end{lemma}

By Lemma \ref{lemma:prob-serverK-explode}, it holds that
$\Pr\{\bar{\bold{Q}} \not \in \set{E}_K\} \leq \frac{208\tau_{1K}b^2}{\hat{\beta}\epsilon N}.$ Then it is natural to discuss the drift of $V_3(\bold{q})$ when it is greater than $B_3$ by conditioning on whether $\bold{q}$ is in $\set{E}_K$ or not. The result is summarized in this lemma, and the proof is in the appendix.
\begin{lemma} \label{lemma:serverM-negativeDrift}
When $V_3(\bold{q}) \geq B_3$, it holds that
\begin{itemize}
	\item if $\bold{q} \in \set{E}_K$, the drift is bounded as $GV_3(\bold{q}) \leq -\frac{B_3\mu_M}{b}+\tilde{d}_2$;
	\item if $\bold{q} \not \in \set{E}_K$, the drift is bounded as $GV_3(\bold{q}) \leq \mu_1.$
\end{itemize}
\end{lemma}
We now apply Lemma \ref{lemma:geometric-tail}. Under the notation of that lemma, it holds $\nu_{\max} = \frac{1}{N}, f_{\max} \leq \mu_1$ for $V_3(\bold{q})$. Let $\gamma \coloneqq\frac{B_3\mu_M}{b}-\tilde{d}_2$, and take $j_3 = \frac{2\mu_1\ln(N)}{\gamma}$. Applying Lemma \ref{lemma:geometric-tail} and using Lemma \ref{lemma:serverM-negativeDrift}, it satisfies that 
\begin{equation}
\Pr\left\{V_3(\bar{\bold{Q}}) > B_3 + \frac{2j_3}{N}\right\} \leq \left(1+\frac{\gamma}{\mu_1}\right)^{-j_3}+\left(\frac{\mu_1}{\gamma}+1\right)\Pr\{\bold{q} \not \in \set{E}_K\}
\leq N^{-2}+\frac{416\mu_1\tau_{1K}b^2}{\beta \epsilon N}
\end{equation}
where the last inequality is because $\gamma < \mu_1$ when $N$ is sufficiently large. Furthermore, the expecation of $V_3(\bar{\bold{Q}})$ can be bounded as 
\begin{align}
\expect{V_3(\bar{\bold{Q}})} &\leq \expect{V_3(\bar{\bold{Q}}) \middle| V_3(\bar{\bold{Q}}) \leq B_3 + \frac{2j_3}{N}} + \expect{V_3(\bar{\bold{Q}}) \middle| V_3(\bar{\bold{Q}}) > B_3 + \frac{2j_3}{N}} \Pr\left\{V_3(\bar{\bold{Q}}) > B_3 + \frac{2j_3}{N}\right\} \\
&\leq B_3+\frac{4\mu_1\ln(N)}{\gamma N}+b\left(N^{-2}+\frac{416\mu_1\tau_{1K}b^2}{\beta \epsilon N}\right) \\
&\leq B_3 + \frac{5\mu_1\ln(N)}{\gamma N}+\frac{416\mu_1\tau_{1K}b^3}{\hat{\beta}\epsilon\gamma N}.
\end{align}
The definition of $B_3$ in (\ref{eq:def-B3}) and that of $\gamma$ immediately give the desired result. 
\end{proof}

\subsection{Throughput Guarantee and the Proof of Theorem \ref{thm:finite-bound}}
The next lemma provides a bound on the blocking probability, and thus characterizes the effective throughput of the system. Due to space limitations, the reader is referred to the appendix for the proof.
\begin{lemma}\label{lemma:blocking}
Under Assumptions \ref{as:buffer} and \ref{as:dense}, the probability $p_{\set{B}}$ that an arrival of job is blocked is bounded as 
\begin{equation}\label{eq:bound-lastM-K}
p_{\set{B}} \leq \frac{\tilde{d_2}}{\lambda} + \frac{52\tau_{1K}b^2}{\epsilon N}. \tag{\ref{eq:finite-block}}
\end{equation}
\end{lemma}

Wrapping up above lemmas, we can conclude the proof of Theorem \ref{thm:finite-bound}.
\begin{proof}[Proof of Theorem \ref{thm:finite-bound}]
The first result and third result in Theorem \ref{thm:finite-bound} corresponds to Lemma \ref{lemma: njobs-firstK} and \ref{lemma:blocking}. For the second result, notice that Lemma \ref{lemma: njobs-firstK} implies 
\begin{equation} \label{eq:bound-firstK}
\expect{\sum_{m=1}^K C_m(\bar{\mathbf{Q}})} \leq C^*+\epsilon+\frac{52\tau_{1K}b^2}{\epsilon N}.
\end{equation}
Then combining (\ref{eq:bound-firstK}) and (\ref{eq:bound-lastM-K}) in Lemma \ref{lemma:njobs-lastM-K} and the assumption that $\tilde{d}_2 \leq \frac{\epsilon \mu_K}{2b}$ in Assumption \ref{as:dense}, it holds
\[
\begin{aligned}
\expect{\sum_{m=1}^M C_m(\bar{\mathbf{Q}})} &= \expect{\sum_{m=1}^K C_m(\bar{\mathbf{Q}})} + \expect{\sum_{m=K+1}^M C_m(\bar{\mathbf{Q}})} \\
&\leq C^*+\epsilon+\frac{52\tau_{1K}b^2}{\epsilon N} + \frac{\tilde{d_2}b}{\mu_M}+2\sqrt{\frac{5\tau_{1M}b\ln N}{ N}}+8b^2\sqrt{\frac{26\tau_{1K}\tau_{1M}}{\hat{\beta}\epsilon N}} \\
&\leq C^*+\left(1+\frac{\mu_K}{2\mu_M}\right)\epsilon + 2\sqrt{\frac{5\tau_{1M}b\ln N}{N}}+60b^2\sqrt{\frac{26\tau_{1K}\tau_{1M}}{\hat{\beta}\epsilon N}},
\end{aligned}
\]
which is exactly (\ref{eq:bound-total}).
\end{proof}

\section{Proof of The Random Graph Results}
In this section, we prove Theorem \ref{thm:random-graph}. Since similar proof holds for Theorem \ref{thm:random-graph-uniform}, we provide that proof in the appendix.
\paragraph{\textbf{Proof Sketch}} The result is proved by showing that almost every pair of large enough subsets of $\set{L}, \set{R}$ shares edges between the two sets because of the random graph structure. To show this fact, we first bound the probability that two given subsets are disconnected. Then the union bound concludes the proof since the total number of pairs of subsets is given by $2^{L+N}.$
\subsection{Proof of Theorem \ref{thm:random-graph}}
\begin{proof}
Recall the definition of $p_1,p_2,\tilde{d}_1,\tilde{d}_2$ in Assumption \ref{as:dense}. W.L.O.G., assume $Np_j$ is an integer for $j = 1,2$. Otherwise, we can raise $p_j$ to satisfy this condition since the size of a subset must be an integer. Suppose that we generate a bipartite graph $G$ as in Theorem \ref{thm:random-graph}. Let $\set{C}_j$ be the event that $G$ violates the $j-$th condition in Assumption \ref{as:dense}. We bound $\Pr\{\set{C}_j\}$ separately. To simplify the notation, let us denote $\set{R}^1 = \set{R}_{K-1}, \set{R}^2 = \set{R}_K$. And let us write $p_{\ell,r}$ be the probability that a port $\ell$ connects with a server $r$ in the graph $G$.

First, define $\set{D}_{\set{K},\set{I}}$ as the event that a subset $\set{K}$ of $\set{L}$ has no edges with a subset $\set{I}$ of $\set{R}$. Then for $j = 1,2$,  
\begin{equation}
\set{C}_j = \bigcup_{\substack{\set{K}\subseteq \set{L} \colon \sum_{\ell \in \set{K}} \lambda_{\ell} > N\tilde{d}_j \\ \set{I}\subseteq \set{R}^j \colon |\set{I}| \geq Np_j}} \set{D}_{\set{K},\set{I}}.
\end{equation}
Fix $j \in \{1,2\}$. Let $\set{K}$ be any subset of $\set{L}$ satisfying $\sum_{\ell \in \set{K}} \lambda_{\ell} > N\tilde{d}_j$, and $\set{I}$ be any subset of $\set{R}^j$ satisfying $|\set{I}| \geq Np_j$. We want to bound $\Pr\{\set{D}_{\set{K},\set{I}}\}$. Notice that by Assumption \ref{as:dense}, it holds $p_1 < p_2, \tilde{d}_1 < \tilde{d}_2$, and $\frac{\tilde{d}_2}{H_2} \geq \frac{\tilde{d}_1}{H_1}$. Then by the construction of $G$, if there is a port $\ell$ in $\set{K}$ such that $\lambda_{\ell} \geq N\tilde{d_j}{H_j}$, this port must be connected to all servers in $\set{R}^j$, meaning that $\Pr\{\set{D}_{\set{K},\set{I}}\} = 0$. Therefore, we can assume that such port does not exist. Recall that $z_{\ell,r}$ is the probability that port $\ell$ is connected with server $r$. It holds that 
\begin{equation}
\Pr\{\set{D}_{\set{K},\set{I}}\} = \prod_{\ell \in \set{K}}\prod_{r \in \set{I}} (1-z_{\ell, r}) 
\leq \exp\left(-\sum_{\ell \in \set{K}}\sum_{r \in \set{I}} z_{\ell, r}\right) 
\leq \exp\left(-\sum_{\ell \in \set{K}}\sum_{r \in \set{I}} \frac{\lambda_{\ell}H_j}{N\tilde{d}_j}\right),
\end{equation}
and thus
\begin{equation}
\Pr\{\set{D}_{\set{K},\set{I}}\} \leq  \exp\left(-|\set{I}|\frac{\sum_{\ell \in \set{K}} \lambda_{\ell}H_j}{N\tilde{d_j}}\right) \\
\leq \exp(-H_jNp_j)
\leq 2^{-2(N+L)}.
\end{equation}
The first inequality is because $\ln(1+x) \leq x$ for $x > -1$, and $z_{\ell,r} < 1$. The second inequality is from the construction of $G$. The third inequality is from the definition of $\set{K}$ and $\set{I}$. It thus holds that $\Pr\{\set{C}_j\} \leq 2^{N+L}2^{-2(N+L)} = 2^{-(N+L)}$ by the union bound. Use the union bound once again, it holds $\Pr\{\set{C}_1 \cup \set{C}_2\} \leq 2^{-(N+L-1)}$.

For the total number of edges used in $G_N$, recall the definition of $p_1,p_2,\tilde{d}_1,\tilde{d}_2$ for a particular system in Assumption \ref{as:dense}, and $H_1,H_2$ in Theorem \ref{thm:random-graph}. It holds that $\frac{\tilde{d}_1}{H_1} = O(\frac{\epsilon^2}{b^5(N+L)/N})$, and $\frac{\tilde{d}_2}{H_2} = O(\frac{\epsilon^2}{b^5(N+L)/N})$. Note that there are four types of connections on graph $G_N$ as per Theorem \ref{thm:random-graph}, we bound their numbers of edges separately. First, the number of ports with $\lambda_{\ell} \geq N\frac{\tilde{d}_1}{H_1}$ is bounded by $\frac{N\mu_1H_1}{N\tilde{d}_1} = O(\frac{b^5(N+L)}{\epsilon^2}N)$ because $\lambda_{\Sigma} \leq N\mu_1$. Therefore, the number of connections from them is bounded by $O(\frac{b^5(N+L)}{\epsilon^2})$ since there are $N$ servers. The same result holds for ports with $\lambda_{\ell} \geq N\frac{\tilde{d}_2}{H_2}$. Now for the remaining ports, the expected number of edges is upper bounded by 
$
2\sum_{\ell \in \set{L}}\frac{\lambda_{\ell}}{N}\left(\frac{H_1}{\tilde{d}_1}+\frac{H_2}{\tilde{d}_2}\right)N = O\left(\frac{b^5(N+L)}{\epsilon^2}\right).
$
Then to sum up, the expected number of edges in $G_N$ scales as $O\left(\frac{b^5(N+L)}{\epsilon^2}\right)$.
\end{proof}

\section{Simulation Results}
In this section, we present simulation results for JFSQ and JFIQ. In particular, the following two settings are explored:
\begin{itemize}
\item we compare the mean response time of JFSQ, JFIQ with a recent paper \cite{gardner2020scalable} in a fixed-size system;
\item we study the convergence of JFSQ and JFIQ on a random bipartite graph in the many-server regime.
\end{itemize}
We will also compare our policies with JSQ and JIQ where we assume that ties in those policies are broken at random. Detailed results are as follows.
\subsection{Performance in a Fixed-Size System}
We first study one particular setting as in \cite{gardner2020scalable}. There are $100$ servers with fast service rate $\frac{25}{9}$, and $400$ servers with slow service rate $\frac{5}{9}$. Jobs arrive into the system in a Poisson process of rate $\lambda_{\Sigma}$, and can be routed to any server. We simulate an infinite buffer system by setting the buffer size at each server to $10^6$. We compare JFSQ and JFIQ with JSQ, JIQ and JSQ-(2,2) introduced in \cite{gardner2020scalable}. JSQ-(2,2) is similar to Pod, and it is shown in \cite{gardner2020scalable} to perform better than other algorithms in light traffic. We refer the reader to the appendix for a detailed description of JSQ-(2,2). Beside, the lower bound result in Theorem \ref{thm:lower-bound} is plotted as a baseline. Define the system load to be $\frac{\lambda_{\Sigma}}{500}$. By increasing the system load, we can obtain Fig. \ref{fig:traffic-range}.
\begin{figure}
 \centering
 \scalebox{0.45}{\input{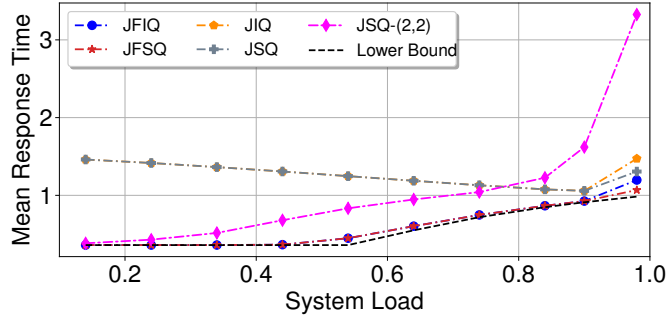}}
 \caption{The Mean Response Time of Different Routing Policies in a Fixed-Size System with Increasing System Load}
 \label{fig:traffic-range}
\end{figure}
Clearly, Fig. \ref{fig:traffic-range} shows that JFSQ and JFIQ can achieve consistently fast mean response (very close to the lower bound) ranging from light traffic to heavy traffic (the system load is around $0.98$). For other policies, JSQ-(2,2) performs well in light traffic. However, JIQ and JSQ could have relatively poor response time in light traffic, although JIQ is shown to have asymptotically zero waiting time \cite{stolyar_2015}.

\subsection{Convergence in the Many-Server Regime}\label{sec:random-manyserver}
Next we explore the convergence behavior of JFSQ and JFIQ when there are job-server constraints. In particular, suppose there are $N$ servers in the system. We assume there are four types of servers with the same amount of each type. The service time distributions are all exponentially distributed, but with different service rate such that $\mu_i = 2^{-i+1}, i=1,2,3,4.$ We also study the convergence of JSQ and JIQ. JSQ-(2,2) introduced above is not studied because it is designed for systems with two classes of servers. 

The number of ports is set as $L = N^{1.5}$. The arrival rate to each port is assumed to be homogeneous, and is equal to $\frac{\lambda_{\Sigma}}{L}$ with $\lambda_{\Sigma} = 0.9\sum_{i=1}^4 \frac{N\mu_i}{4}$. Denote the system load as $\lambda = 0.9$. In the corresponding bipartite graph, each port connects with each server with probability $\frac{2\sqrt{\ln N}}{N(1-\lambda)}\ln \frac{1}{1-\lambda}$ according to Theorem \ref{thm:random-graph-uniform}. The buffer size in this case is set as $b=5$ because in many-server systems, we expect there to be little queueing and one should not need a large buffer size. Fig. \ref{fig:random} presents the convergence behavior of the mean-response time for JFSQ, JFIQ, JIQ and JSQ.  
\begin{figure}
	\centering
	\scalebox{0.45}{\input{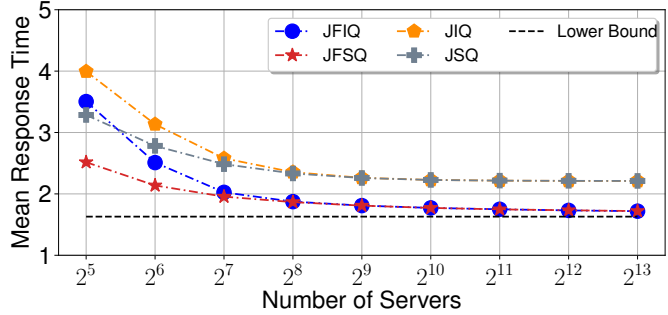}}
	\caption{The Mean Response Time of Different Routing Policies on Increasing-Sized Random Bipartite Graphs}
	\label{fig:random}
\end{figure}
It is interesting to notice that both JIQ and JFIQ suffer from slow mean response time when the system is small. But when the number of servers is $2^{11} = 2048$, the mean response time of JFSQ and JFIQ is very close to the lower bound. Such requirement on the number of servers is fine since modern cloud platforms can easily possess tens of thousands of servers \cite{amvrosiadis_2018}. On the other hand, both JSQ and JIQ also converge as $N$ increases. Nevertheless, their mean response time is not optimal because they neglect server heterogeneity. Note that when the system is large, the blocking probability is nearly zero, even with a small buffer size. The convergence of the blocking probability is provided in the appendix. The setting is also extended to hyper-exponential service time distribution. For this new distribution, we show that although JFSQ and JFIQ have slow mean response times initially, their convergence behavior is similar to Fig. \ref{fig:random} when $N$ increases. We refer the reader to the appendix for details. 
\section{Conclusion}
In this paper, we studied the performance of two load balancing policies, JFSQ and JFIQ for load balancing on a bipartite graph. For a "well-connected" bipartite graph, we presented a bound on the mean response time for finite-size systems, which implies asymptotic optimality in the mean response time in both the many-server regime and the sub Halfin-Whitt regime. A by-product of this paper is a novel technique for bouding the distance to the mean-field limit of heterogeneous load balancing systems. In the analysis, we established three state-space collapse results to show that the system behaves similar to its mean-field limit. We also presented how to construct a sparse "well-connected" bipartite graph, where each left node is only connected to $\omega(\frac{1}{(1-\lambda)^2})$ right nodes when arrival rates are heterogeneous, and only $\omega(\frac{1}{1-\lambda}\ln \frac{1}{1-\lambda})$ nodes for homogeneous servers, given that the buffer size is a constant, and the number of left nodes is at least that of right nodes. However, it is unknown whether these two bounds are tight, which we leave for future research. 

\paragraph{Acknowledgment:} The work of Wentao Weng was conducted during a visit to the Coordinated Science Lab, UIUC during 2020.
\bibliographystyle{abbrvnat}
\bibliography{wentao-references}

\appendix
\section{Proof of Proposition \ref{thm:lower-bound}}
\renewcommand{\thmnamerestated}{Proposition~\ref{thm:lower-bound}[Restated]}
\newtheorem*{thmrestatedfilling}{\thmnamerestated}{\bfseries}{\itshape}
\begin{thmrestatedfilling}
Suppose that the buffer size is infinite, i.e. $b = \infty$. Let $\bar{Z}$ be the random variable denoting the service time of one job. Then for any stable policy, the mean number of jobs in the system is lower bounded by $NC^*$, and 
\begin{equation}
\expect{\bar{Z}} \geq \frac{C^*}{\lambda}.
\end{equation}
\end{thmrestatedfilling}
\begin{proof}
For any $m \in \{1,\cdots,M\}$, let $I_m$ denote the probability that an arrival of jobs is scheduled to a type-$m$ server in steady state. Also, recall that $\bar{s}_{m,1}$ is defined as a steady-state random variable denoting the number of busy type-$m$ servers divided by $N$. Then because of stability and work conservation law, it holds that for all $m \leq M$,
\begin{equation}\label{eq:flow-balance-each}
\lambda_{\Sigma}I_m = N\mu_m\expect{\bar{S}_{m,1}}.
\end{equation}
In particular, 
\begin{equation}\label{eq:flow-balance-total}
\lambda = \sum_{m=1}^M \frac{\lambda_{\Sigma}I_m}{N} = \sum_{m=1}^M \mu_m\expect{\bar{S}_{m,1}}
\end{equation}
since $\sum_{m=1}^M I_m = 1$. Now notice that the mean service time of jobs is given by 
\begin{equation}
\expect{\bar{Z}} =\sum_{m=1}^M \frac{I_m}{\mu_m}=\sum_{m=1}^M \frac{\expect{\bar{S}_{m,1}}}{\lambda}
\end{equation}
since the service time at type-$m$ servers is exponentially distributed with mean $\frac{1}{\mu_m}$, and $I_m$ satisfies (\ref{eq:flow-balance-each}). To obtain a lower bound of $\expect{\bar{Z}}$, consider the following linear programming.
\begin{align*}
	\min\quad        & \frac{1}{\lambda}\sum_{m=1}^M x_m   \\
	\text{s.t.\quad} & 
	\lambda=\sum_{m=1}^M \mu_m x_m,~~ m=1,\dots,M \\
	& 0 \leq x_m \leq \alpha_m,~~m=1,\dots,M 
\end{align*}
where $x_m$ is an analog of $\expect{\bar{S}_{m,1}}$, and the objective value is a lower bound of $\expect{\bar{Z}}$ because of (\ref{eq:flow-balance-total}). Then since only the sum of $x_m$ matters, and $\mu_1 \geq \cdots \geq \mu_M$, the optimal solution is exactly given by $x^*_1 = \alpha_1,\cdots,x^*_{K-1}=\alpha_{K-1},x^*_K=\frac{\lambda-\sum_{m=1}^{K-1}\mu_m x_m}{\mu_K}, x^*_m = 0$ for $m > K$. Then it is clear that $\expect{\bar{Z}} \geq \frac{1}{\lambda}\sum_{m=1}^M x^*_m = \frac{C^*}{\lambda}.$
\end{proof}

\section{Proof of Lemmas in Section \ref{sec:upper-bound}}

\subsection{Proof of Lemma \ref{lemma:zero-drift}}
\renewcommand{\thmnamerestated}{Lemma~\ref{lemma:zero-drift}[Restated]}
\begin{thmrestatedfilling}
The expectation $\expect{Gg(\sum_{m=1}^K C_m(\bar{\mathbf{Q}}))}$ is equal to $0$.
\end{thmrestatedfilling}
\begin{proof}
To simplify the notation, denote $V(\bold{q}) = g(\sum_{m=1}^K C_m(\bold{q}))$ for a state $\bold{q}$. Now that since the system is stable (because of the assumption of finite buffers), there is a unique stationary distribution $\pi_{\bold{q}}$ that solves the balancing equation such that for every $\bold{q}$,
\begin{equation}
\pi_{\bold{q}}\sum_{\bold{q'}} r_{\bold{q},\bold{q'}}=\sum_{\bold{q'}} \pi_{\bold{q}'}r_{\bold{q'},\bold{q}}
\end{equation}
where $r_{\bold{q},\bold{q'}}$ is the transition rate from $\bold{q}$ to $\bold{q'}$.
Now that $V(\bold{q})$ is bounded (as $\sum_{m=1}^K C_m(\bold{q} \leq b$), it holds
\[
\begin{aligned}
\expect{GV(\bar{\bold{Q}})} &= \sum_{\bold{q}} \pi_{\bold{q}}\sum_{\bold{q'}} r_{\bold{q},\bold{q'}}(V(\bold{q'})-V(\bold{q})) \\
&= -\sum_{\bold{q}}\pi_{\bold{q}}\sum_{\bold{q'}} V(\bold{q})r_{\bold{q},\bold{q'}} + \sum_{\bold{q}}\pi_{\bold{q}}\sum_{\bold{q}'}r_{\bold{q},\bold{q'}}V(\bold{q'}) \\
&= -\sum_{\bold{q}} V(\bold{q})\sum_{\bold{q'}} \pi_{\bold{q}}r_{\bold{q},\bold{q'}} + \sum_{\bold{q}} V(\bold{q})\sum_{\bold{q'}} \pi_{\bold{q'}}r_{\bold{q'},\bold{q}} \\
&= 0. 
\end{aligned}
\]
\end{proof}

\subsection{Proof of Lemma \ref{lemma:firstK-key-term}}
\renewcommand{\thmnamerestated}{Lemma~\ref{lemma:firstK-key-term}[Restated]}
\begin{thmrestatedfilling}
It holds that
\begin{equation}
\expect{h\left(\sum_{m=1}^K C_m(\bar{\mathbf{Q}})\right)} \leq \expect{\mathbbm{1}\left\{\sum_{m=1}^K C_m(\bar{\bold{Q}}) \geq \eta + \frac{1}{N}\right\} g'\left(\sum_{m=1}^K C_m(\bar{\bold{Q}})\right)(\lambda+\mu_1\delta-W(\bar{\bold{Q}}))} + \frac{38b^2\tau_{1K}}{\epsilon N}. \tag{\ref{eq:refined-key-term}}
\end{equation}
\end{thmrestatedfilling}
\begin{proof}
The idea is to utilize the result that $\expect{h\left(\sum_{m=1}^K C_m(\bar{\mathbf{Q}})\right)} \leq (\ref{eq:initial-value}) (\ref{eq:arrival-term})+(\ref{eq:departure-term})$, and to expand (\ref{eq:arrival-term}) and (\ref{eq:departure-term}) by Taylor's expansion. Consider three cases of state $\bold{q}$. 
\begin{itemize}
\item First, if $\sum_{m=1}^K C_m(\bold{q}) \leq \eta - \frac{1}{N}$, then $g(\sum_{m=1}^K C_m(\bold{q})-\frac{1}{N}), g(\sum_{m=1}^K C_m(\bold{q})), g(\sum_{m=1}^K C_m(\bold{q}) + \frac{1}{N})$ are all zero. This case has no contribution to the expectation;
\item second, if $\sum_{m=1}^K C_m(\bold{q}) \in (\eta - \frac{1}{N},\eta+\frac{1}{N})$, by first-order Taylor's expansion, there exists some $\tilde{\xi}_{\bold{q}},\tilde{\eta}_{\bold{q}} \in (\eta - \frac{2}{N},\eta+\frac{2}{N})$, such that 
\[
\begin{aligned}
g\left(\sum_{m=1}^K C_m(\bold{q}) + \frac{1}{N}\right) - g\left(\sum_{m=1}^K C_m(\bold{q})\right) &= \frac{1}{N}g'(\tilde{\xi}_{\bold{q}}), \\
g\left(\sum_{m=1}^K C_m(\bold{q}) - \frac{1}{N}\right) - g\left(\sum_{m=1}^K C_m(\bold{q})\right) &= \frac{-1}{N}g'(\tilde{\eta}_{\bold{q}}); 
\end{aligned}
\]
\item third, if $\sum_{m=1}^K C_m(\bold{q}) \geq \eta+\frac{1}{N}$, by second-order Taylor's expansion, there exists some $\xi_{\bold{q}},\eta_{\bold{q}}$, such that 
\[
\begin{aligned}
g\left(\sum_{m=1}^K C_m(\bold{q}) + \frac{1}{N}\right) - g\left(\sum_{m=1}^K C_m(\bold{q})\right) &= \frac{1}{N}g'\left(\sum_{m=1}^K C_m(\bold{q})\right) + \frac{2}{N^2}g''(\xi_{\bold{q}}), \\
g\left(\sum_{m=1}^K C_m(\bold{q}) - \frac{1}{N}\right) - g\left(\sum_{m=1}^K C_m(\bold{q})\right) &= -\frac{1}{N}g'\left(\sum_{m=1}^K C_m(\bold{q})\right) + \frac{2}{N^2}g''(\eta_{\bold{q}}). 
\end{aligned}
\]

\end{itemize} 
Then it holds that 
\begin{align}
&\mspace{20mu}\expect{h\left(\sum_{m=1}^K C_m(\bar{\mathbf{Q}})\right)} \\
 &\leq (\ref{eq:initial-value})+ (\ref{eq:arrival-term})+(\ref{eq:departure-term}) \\
&= \expect{\mathbbm{1}\left\{\sum_{m=1}^K C_m(\bar{\bold{Q}}) \geq \eta + \frac{1}{N}\right\}\left(g'\left(\sum_{m=1}^K C_m(\bar{\bold{Q}})\right)(\lambda+\mu_1\delta-W(\bar{\bold{Q}}))\right)} \\
&\mspace{20mu}+\expect{\mathbbm{1}\left\{\sum_{m=1}^K C_m(\bar{\bold{Q}}) \geq \eta + \frac{1}{N}\right\}\left(\frac{2}{N}\left(\lambda g''(\xi_{\bar{\bold{Q}}})+W(\bar{\bold{Q}})g''(\eta_{\bar{\bold{Q}}})\right)\right)} \label{eq:second-gradient-bound}\\
&\mspace{20mu}+\expect{\mathbbm{1}\left\{\sum_{m=1}^K C_m(\bar{\bold{Q}}) \in (\eta-\frac{1}{N},\eta + \frac{1}{N})\right\}\left(g'\left(\sum_{m=1}^K C_m(\bar{\bold{Q}})\right)(\mu_1\delta)+\lambda g'(\tilde{\xi}_{\bar{\bold{Q}}})-W(\bar{\bold{Q}})g'(\tilde{\eta}_{\bar{\bold{Q}}})\right)} \label{eq:first-gradient-bound}.
\end{align}
It suffices to bound (\ref{eq:second-gradient-bound}) and (\ref{eq:first-gradient-bound}). First, note that $|g''(x)| \leq \frac{1}{\mu_1\delta}$ for all $x$ by the explicit form of $g(x)$ in (\ref{eq:g-gradient}). It holds
\begin{equation}
(\ref{eq:second-gradient-bound}) \leq \frac{2}{N}\cdot\frac{1}{\mu_1\delta}\cdot 2\mu_1 = \frac{4}{N\delta} = \frac{24\tau_{1K}b^2}{\epsilon N}.
\end{equation}
On the other hand, to bound (\ref{eq:first-gradient-bound}), since $\sum_{m=1}^K C_m(\bar{\bold{Q}}), \tilde{\xi}_{\bold{q}},\tilde{\eta}_{\bold{q}} \in (\eta - \frac{2}{N},\eta+\frac{2}{N})$, their derivatives are all bounded by $\frac{2}{N\mu_1\delta}.$ Then
\begin{equation}
(\ref{eq:first-gradient-bound}) \leq \frac{2}{N\mu_1\delta}\cdot(\mu_1\delta+\mu_1)=\frac{2}{N}+\frac{12\tau_{1K}b^2}{\epsilon N} \leq \frac{14\tau_{1K}b^2}{\epsilon N}.
\end{equation}
Summing the above two equations completes the proof of Lemma \ref{lemma:firstK-key-term}.
\end{proof}

\subsection{Proof of Lemma \ref{lemma:lyapunov-firstK-1}}
\renewcommand{\thmnamerestated}{Lemma~\ref{lemma:lyapunov-firstK-1}[Restated]}
\begin{thmrestatedfilling}
Consider the following Lyapunov function
\begin{equation}
V_1(\bold{q}) = \min\left(\sum_{j=1}^b s_{K,j}(\bold{q})+\sum_{m=1}^{K-1}\sum_{j=2}^b s_{m,j}(\bold{q}), \sum_{m=1}^{K-1} C_m^*-\sum_{m=1}^{K-1} s_{m,1}(\bold{q})\right). \tag{\ref{eq:lyapunov-firstK-1}}
\end{equation}
It holds that if $V_1(\bold{q}) \geq B_1 \coloneqq \tau_{1K}\delta$, then $GV_1(\bold{q}) \leq \frac{-\mu_1\delta}{2b}$.
\end{thmrestatedfilling}

\begin{proof}
Since $V_1(\bold{q}) \geq B_1$ by assumption, both of the following two properties holds:
\begin{align}
\sum_{j=1}^b s_{K,j}(\bold{q})+\sum_{m=1}^{K-1}\sum_{j=2}^b s_{m,j}(\bold{q}) &\geq B_1  \label{eq:V1-largeQ}; \\
\sum_{m=1}^{K-1} s_{m,1}(\bold{q}) &\leq \sum_{m=1}^{K-1} C_m^* - B_1. \label{eq:V1-idleK}
\end{align}
Let $\set{T}_{1,1}$ be the first term in $V_1(\bold{q})$, and $\set{T}_{1,2}$ be the second term. First, by definition, 
\begin{align}
GV_1(\bold{q}) &= \sum_{\bold{q'}} r_{\bold{q},\bold{q'}} \left(V_1(\bold{q'})-V_1(\bold{q})\right) \notag \\
&= \sum_{\bold{q'},\text{arrival}} r_{\bold{q},\bold{q'}} \left(V_1(\bold{q'})-V_1(\bold{q})\right) \label{eq:V1-arrival} \\
&\mspace{20mu}+\sum_{\bold{q'},\text{departure}} r_{\bold{q},\bold{q'}} \left(V_1(\bold{q'})-V_1(\bold{q})\right) \label{eq:V1-departure}
\end{align}
where we separate transitions by identifying those caused by a job arrival from those caused by a job departure. Bounding (\ref{eq:V1-arrival}) and (\ref{eq:V1-departure}) can then bound $GV_1(\bold{q})$. Next we consider two cases corresponding to whether $V_1(\bold{q})$ is equal to $\set{T}_{1,1}$ or to $\set{T}_{1,2}$.

Suppose that $\set{T}_{1,1} \leq \set{T}_{1,2}$. then in this case,
\begin{align}
(\ref{eq:V1-departure}) &\leq -\left(\sum_{j=1}^b \mu_K(s_{K,j}(\bold{q})-s_{K,j+1}(\bold{q}))+\sum_{m=1}^{K-1}\sum_{j=2}^b\mu_m(s_{m,j}(\bold{q})-s_{m,j+1}(\bold{q}))\right) \label{eq:V1-case1-dep-e1}\\
&= -\left(\mu_Ks_{K,1}(\bold{q})+\sum_{m=1}^K \mu_ms_{m,2}(\bold{q})\right)\label{eq:V1-case1-dep-e2} \\
&\leq -\frac{B_1\mu_K}{b} \leq \frac{-\mu_1\delta}{b}.\label{eq:V1-case1-dep-e3}  
\end{align}
The first inequality (\ref{eq:V1-case1-dep-e1}) is because $V_1(\bold{q}) = \tau_{1,1}$, and only jobs departing from servers of type $K$ and servers of types less than $K$ with queue length at least $2$ can affect the value of $V_1(\bold{q})$. The first equation (\ref{eq:V1-case1-dep-e2}) comes from the fact that $s_{m,b+1}=0$ for all $m$. The last inequality is from (\ref{eq:V1-largeQ}) and the non-decreasing property 
\[
s_{m,1}(\bold{q}) \geq s_{m,2}(\bold{q}) \geq \cdots s_{m,b}(\bold{q}) 
\]
for all $m$.

On the other hand, to bound (\ref{eq:V1-arrival}), notice that $V_1(\bold{q})$ can increase only when a job arrival is routed to some servers of types at least $K$. Then clearly,
\begin{equation}
(\ref{eq:V1-arrival}) \leq \sum_{\ell=1}^L \frac{1}{N}\lambda_{\ell}\cdot \mathbbm{1}\left\{\text{an arrival to port }\ell \text{ is not routed to an idle server of types less than }k \mid \bold{q}\right\}. \label{eq:V1-case1-arrival-e1}
\end{equation}
However, by (\ref{eq:V1-idleK}), the number of idle servers of types less than $K$ is at least
\[
N\sum_{m=1}^{K-1} \left(C_m^*-s_{m,1}(\bold{q})\right) \geq NB_1 = \frac{N\epsilon}{6b^2}.
\]
Let $\set{I}$ be the set of idle servers of types less than $K$. Since $|\set{I}| \geq \frac{N\epsilon}{6b^2}$, Assumption \ref{as:dense} guarantees that $\sum_{\ell \not \in N_{\set{R}}(\set{I})} \lambda_{\ell} \leq N\tilde{d}_1 = \frac{N\epsilon \mu_K}{12b^3}.$ That is to say, the total arrival rates of ports not connected with servers in $\set{I}$ is bounded by $N\tilde{d}_1$. Now since our routing policy is either JFSQ or JFIQ, for those ports connected with $\set{I}$, a job arrival must be routed to one server in $\set{I}$ because servers in $\set{I}$ are idle, and are faster than other idle servers not in $\set{I}$. Therefore, 
\begin{equation}
(\ref{eq:V1-case1-arrival-e1}) \leq \frac{1}{N}\cdot \frac{N\epsilon \mu_K}{12b^3} \leq \frac{\mu_1\delta}{2b}. \label{eq:V1-case1-arrival-final}
\end{equation}
With (\ref{eq:V1-case1-dep-e3}) and (\ref{eq:V1-case1-arrival-final}), it holds $GV_1(\bold{q}) \leq \frac{-\mu_1\delta}{2b}$ when $\set{T}_{1,1} \leq \set{T}_{1,2}$.

For the second case where $\set{T}_{1,1} \geq \set{T}_{1,2}$, it holds
\begin{equation}
(\ref{eq:V1-departure}) \leq \sum_{m=1}^{K-1} \mu_m\left(s_{m,1}(\bold{q})-s_{m,2}(\bold{q})\right)
\end{equation}
since $V_1(\bold{q})$ increases only when a job departs from a server of type less than $K$ and only with this single job in the server. Also, we can see
\begin{align}
(\ref{eq:V1-arrival}) &\leq -\frac{1}{N}\sum_{\ell=1}^L \lambda_{\ell}\cdot \mathbbm{1}\left\{\text{an arrival to port }\ell \text{ is routed to an idle server of type less than }k \mid \bold{q}\right\} \\
&\leq \frac{1}{N}(-\lambda_{\Sigma} + N\tilde{d}_1) = -\lambda + \tilde{d}_1.
\end{align}
The first inequality is because for arrival transitions, only jobs arriving to idle servers of types less than $k$ can change $V_1(\bold{q})$, and their arrivals will all decrease $V_1(\bold{q})$ by $\frac{1}{N}$ by the definition of $\set{T}_{1,2}$. The second inequality is derived from the same argument of (\ref{eq:V1-case1-arrival-final}). Therefore, it holds that
\begin{align}
GV_1(\bold{q}) = (\ref{eq:V1-arrival})+(\ref{eq:V1-departure}) \leq -\lambda + \tilde{d_1}+\sum_{m=1}^{K-1} \mu_m\left(s_{m,1}(\bold{q})-s_{m,2}(\bold{q})\right)
&\leq -\lambda + \tilde{d_1}+\sum_{m=1}^{K-1}\mu_m\alpha_m - \mu_K B_1 \\
&\leq -\mu_KB_1 + \tilde{d_1} \\
&\leq -\frac{\mu_1\delta}{2b}
\end{align}
because of (\ref{eq:V1-idleK}) and the assumption that $\lambda \geq \sum_{m=1}^{K-1} \mu_m\alpha_m.$

Therefore, the above discussion proves that whenever $V_1(\bold{q}) \geq B_1$, it holds $GV_1(\bold{q}) \leq -\frac{\mu_1\delta}{2b}.$
\end{proof}

\subsection{Proof of Lemma \ref{lemma:lyapunov-firstK}}

\renewcommand{\thmnamerestated}{Lemma~\ref{lemma:lyapunov-firstK}[Restated]}
\begin{thmrestatedfilling}
Consider the following Lyapunov function
\begin{equation}
	V_2(\bold{q}) = \min\left(\sum_{m=1}^K \sum_{j=2}^b s_{m,j}(\bold{q}), \sum_{m=1}^K C_m^* + B_2 + 3\tau_{1K}\bar{\delta}-\sum_{m=1}^K s_{m,1}(\bold{q})\right) \tag{\ref{eq:lyapunov-firstK}}
\end{equation}
where $\bar{\delta} \coloneqq \tau_{1K}\delta$, and $B_2 \coloneqq \frac{1}{2}\epsilon + \bar{\delta}.$ It holds that if $V_2(\bold{q}) \geq B_2$, then $GV_2(\bold{q}) \leq -\frac{\mu_1\delta}{b}.$
\end{thmrestatedfilling}
\begin{proof}
Let $\set{T}_{2,1}$ be the first term in $V_2(\bold{q})$, and $\set{T}_{2,2}$ be the second term. Since $V_2(\bold{q}) \geq B_2$, both the following hold:
\begin{align}
\sum_{m=1}^K \sum_{j=2}^b s_{ij}(\bold{q}) &\geq B_2 \label{eq:V2-largeQ};\\
\sum_{m=1}^K s_{m,1}(\bold{q}) &\leq \sum_{m=1}^K C_m^i + 3\mu\bar{\delta} \label{eq:V2-idleK}.
\end{align}
By definition,
\begin{align}
GV_2(\bold{q}) &= \sum_{\bold{q'},\text{arrival}} r_{\bold{q},\bold{q'}}\left(V_2(\bold{q'})-V_2(\bold{q})\right) \label{eq:V2-arrival} \\
&\mspace{20mu}+\sum_{\bold{q'},\text{departure}} r_{\bold{q},\bold{q'}}\left(V_2(\bold{q'})-V_2(\bold{q})\right) \label{eq:V2-departure}. 
\end{align}
We then consider two cases. First, suppose that $\set{T}_{2,1} \leq \set{T}_{2,2}$. Then similar to the proof of Lemma \ref{lemma:lyapunov-firstK-1}, using (\ref{eq:V2-largeQ}), it holds that
\begin{align}
(\ref{eq:V2-departure}) &\leq -\frac{1}{N}\sum_{m=1}^K \sum_{j=2}^b N\mu_m\left(s_{m,j}(\bold{q})-s_{m,j+1}(\bold{q})\right) \\
&= -\frac{1}{N}\sum_{m=1}^K N\mu_ms_{m,2}(\bold{q}) \\
&\leq -\frac{B_2\mu_K}{b} = -\frac{\epsilon \mu_K}{2b}-\frac{\mu_1\delta}{b}.
\end{align}
On the other hand, we have 
\begin{equation}
(\ref{eq:V2-arrival}) \leq \sum_{\ell=1}^L \frac{1}{N}\lambda_{\ell}\cdot \mathbbm{1}\left\{\text{an arrival to port }\ell \text{ is not routed to an idle server of types } \leq k \mid \bold{q}\right\}. \label{eq:V2-case1-arrival-e1}
\end{equation}
Notice that by (\ref{eq:V2-idleK}), the number of idle servers of types no greater than $K$ satisfies that 
\begin{align}
&\mspace{20mu}N\left(\sum_{m=1}^K \alpha_m - \sum_{m=1}^K s_{m,1}(\bold{q})\right) \\
&\geq N\left(\sum_{m=1}^K \alpha_m - \sum_{m=1}^K C_m^* - 3\tau_{1,K}\bar{\delta}\right) \\
&= N\left(\alpha_K - \frac{\lambda - \sum_{m=1}^{K-1} \mu_m \alpha_m}{\mu_K} - 3\tau_{1,K}\bar{\delta}\right) \\
&= N\cdot \frac{\sum_{m=1}^K \mu_m\alpha_m - \lambda}{\mu_K}-3N\tau_{1K}\bar{\delta} \\
&= \frac{N}{\mu_K}\left(\beta\sum_{m=1}^{K} \mu_m \alpha_m - 3\mu_1\tau_{1K}\delta\right) \\
&\geq N\left(\hat{\beta}-3\tau_{1K}\frac{\epsilon}{6b^2}\right) \geq \frac{N\hat{\beta}}{2} \label{eq:V2-case1-arrival-e2}
\end{align}
where (\ref{eq:V2-case1-arrival-e2}) is because $b^2 \geq \tau_{1K}$ by Assumption \ref{as:buffer}, and $\hat{\beta} = \beta\sum_{m=1}^K \alpha_m$, and $\mu_1 > \cdots > \mu_K$. 

Let $\set{I}$ be the set of idle servers of types no greater than $K$. It then holds $|\set{I}| \geq \frac{N\hat{\beta}}{2}$. Then By Assumption \ref{as:dense}, the total arrival rate of ports not connected with $\set{I}$ is bounded by $N\tilde{d}_2$. Since the routing policy is either JFSQ or JIFQ, jobs arriving to ports connecting with $\set{I}$ must be routed to servers in $\set{I}$. Therefore, it holds $(\ref{eq:V2-case1-arrival-e1}) \leq \tilde{d}_2 \leq \frac{\mu_K\epsilon}{2b}$. Then in this case, we know 
\[
GV_2(\bold{q}) = (\ref{eq:V2-arrival}) + (\ref{eq:V2-departure}) \leq -\frac{\epsilon \mu_K}{2b}-\frac{\mu_1\delta}{b}+\frac{\mu_K\epsilon}{2b} \leq -\frac{\mu_1\delta}{b}.
\] 
Now we consider the second case, $\set{T}_{2,1} \geq \set{T}_{2,2}$. Similarly, it holds
$
(\ref{eq:V2-departure}) \leq \sum_{m=1}^K \mu_m\left(s_{m,1}(\bold{q})-s_{m,2}(\bold{q})\right),
$
and 
\begin{equation}
\begin{aligned}
(\ref{eq:V2-arrival}) &\leq -\frac{1}{N}\sum_{\ell=1}^L \frac{1}{N}\lambda_{\ell}\cdot \mathbbm{1}\left\{\text{an arrival to port }\ell \text{ is routed to an idle server of types } \leq k \mid \bold{q}\right\} \\
&\leq -\lambda + \tilde{d_2}
\end{aligned}
\end{equation}
where the last inequality follows the same argument as in the first case. Then it holds
\begin{align}
GV_2(\bold{q}) &\leq \sum_{m=1}^K \mu_m s_{m,1}(\bold{q})-\sum_{m=1}^K \mu_m s_{m,2}(\bold{q}) - \lambda + \tilde{d_2} \\
&\leq \sum_{m=1}^{K-1} \mu_m\alpha_m + \mu_K(C_K^*+3\mu_1\bar{\delta})-\lambda-\frac{\mu_KB_2}{b-1}+\frac{\mu_K\epsilon}{2b} \\
&\leq 3\mu_1\delta - \frac{\mu_KB_2}{b-1}+\frac{\epsilon}{2b} \\
&\leq 3\mu_1\delta - \frac{\mu_K\epsilon}{2(b-1)}+\frac{\mu_K\epsilon}{2b}-\frac{\mu_1\delta}{b} \\
&\leq -\frac{\mu_1\delta}{b}.
\end{align}
The last inequality is because 
\[
\frac{\mu_K\epsilon}{2(b-1)}-\frac{\mu_K\epsilon}{2b}=\frac{\mu_K\epsilon}{2b^2} \geq 3\mu_1\frac{\mu_K\epsilon}{6\mu_1b^2} = 3\mu_1\delta.
\]
Therefore, we complete the proof of Lemma \ref{lemma:lyapunov-firstK}.
\end{proof}

\subsection{Proof of Lemma \ref{lemma:prob-serverK-explode}}
\renewcommand{\thmnamerestated}{Lemma~\ref{lemma:prob-serverK-explode}[Restated]}
\begin{thmrestatedfilling}
For any $\Delta \geq \frac{\hat{\beta}}{2}$, it holds $\Pr\{\sum_{m=1}^K C_m(\bar{\bold{Q}}) > C^*+\Delta\} \leq \frac{104\tau_{1K}b^2}{\Delta \epsilon N}.$
\end{thmrestatedfilling}
\begin{proof}
By Lemma \ref{lemma: njobs-firstK}, it holds that 
\begin{align}
\Pr\left\{\sum_{m=1}^K C_m(\bar{\bold{Q}}) > C^*+\Delta\right\} &= \Pr\left\{\sum_{m=1}^K C_m(\bar{\bold{Q}}) - C^* - \frac{\hat{\beta}}{4}> \Delta - \frac{\beta}{4}\right\} \\
&\leq \Pr\left\{\sum_{m=1}^K C_m(\bar{\bold{Q}}) - C^* - \frac{\hat{\beta}}{4}> \frac{1}{2}\Delta\right\} \\
&\leq \frac{\expect{\max\left(\sum_{m=1}^K C_m(\bar{\bold{Q}}) - C^* - \frac{\hat{\beta}}{4}, 0\right)}}{\frac{1}{2}\Delta} \\
&\leq \frac{208\tau_{1K}b^2}{\Delta \epsilon N}
\end{align}
since $\epsilon \leq \frac{\hat{\beta}}{4}$ by assumption.
\end{proof}

\subsection{Proof of Lemma \ref{lemma:serverM-negativeDrift}}
\renewcommand{\thmnamerestated}{Lemma~\ref{lemma:serverM-negativeDrift}[Restated]}
\begin{thmrestatedfilling}
When $V_3(\bold{q}) \geq B_3$, it holds that
\begin{itemize}
	\item if $\bold{q} \in \set{E}_K$, the drift is bounded as $GV_3(\bold{q}) \leq -\frac{B_3\mu_M}{b}+\tilde{d}_2$;
	\item if $\bold{q} \not \in \set{E}_K$, the drift is bounded as $GV_3(\bold{q}) \leq \mu_1.$
\end{itemize}
\end{thmrestatedfilling}
\begin{proof}
By definition, 
\begin{align}
	GV_3(\bold{q}) &= \sum_{\bold{q'}} r_{\bold{q},\bold{q'}} \left(V_3(\bold{q'})-V_3(\bold{q})\right) \notag \\
	&= \sum_{\bold{q'},\text{arrival}} r_{\bold{q},\bold{q'}} \left(V_3(\bold{q'})-V_3(\bold{q})\right) \label{eq:V3-arrival} \\
	&\mspace{20mu}+\sum_{\bold{q'},\text{departure}} r_{\bold{q},\bold{q'}} \left(V_3(\bold{q'})-V_3(\bold{q})\right). \label{eq:V3-departure}
\end{align}
Note that since $V_3(\bold{q}) \geq B_3$, and $V_3(\bold{q})=\sum_{m=K+1}^M \sum_{j=1}^b s_{m,j}(\bold{q})$, it holds that 
\begin{equation}
(\ref{eq:V3-departure}) = -\sum_{m=k+1}^M \mu_ms_{m,1}(\bold{q}) \geq -\frac{B_3\mu_M}{b}
\end{equation}
since $s_{m,1}(\bold{q}) \geq \cdots \geq s_{m,b}(\bold{q})$ and $s_{m,b+1}(\bold{q}) = 0$ for all $m$.

For (\ref{eq:V3-arrival}), we consider two cases. First, if $\bold{q} \in \set{E}_K$, the number of idle servers of types no greater than $K$ is given by
\[
\begin{aligned}
&\mspace{20mu}N\left(\sum_{m=1}^K \alpha_m - \sum_{m=1}^K s_{m,1}(\bold{q})\right) \\
&\geq N\left(\sum_{m=1}^K \alpha_m - \sum_{m=1}^K C_m(\bold{q})\right) \\
&\geq  N\left(\sum_{m=1}^K \alpha_m - C^* - \frac{\hat{\beta}}{2}\right) \\
&= N\left(\frac{\beta\sum_{m=1}^{K-1}\alpha_m\mu_m}{\mu_K}-\frac{\hat{\beta}}{2}\right) \\
&\geq N\frac{\hat{\beta}}{2}
\end{aligned}
\]
where the second inequality is because $sum_{m=1}^K C_m(\bold{q}) \leq C^*+\frac{\hat{\beta}}{2}$ when $\bold{q} \in \set{E}_K$. Then since the routing policy is either JFSQ or JFIQ, jobs arriving to ports connecting with idle servers of types no greater than $K$ must be routed to those servers. And by Assumption \ref{as:dense}, the total arrival rate of disconnected ports is bounded by $\tilde{d}_2 N$. As a result,
\begin{equation}
(\ref{eq:V3-arrival}) \leq \tilde{d}_2,
\end{equation}
showing that $GV_3(\bold{q}) \leq -\frac{B_3\mu_M}{b}+\tilde{d}_2$ when $\bold{q} \in \set{E}_K$. 

When $\bold{q} \not \in \set{E}_K$, it holds that $(\ref{eq:V3-arrival}) \leq \lambda \leq \mu_1$, and $(\ref{eq:V3-departure}) \geq 0$. Therefore, $GV_3(\bold{q}) \leq \mu_1$.
\end{proof}

\subsection{Proof of Lemma \ref{lemma:blocking}}
\renewcommand{\thmnamerestated}{Lemma~\ref{lemma:serverM-negativeDrift}[Restated]}
\begin{thmrestatedfilling}
Under Assumption \ref{as:buffer} and Assumption \ref{as:dense}, the probability $p_{\set{B}}$ that an arrival of job is blocked is bounded as 
	\begin{equation}
	p_{\set{B}} \leq \frac{\tilde{d_2}}{\lambda} + \frac{52\tau_{1K}b^2}{\epsilon N}. \tag{\ref{eq:finite-block}}
	\end{equation}
\end{thmrestatedfilling}
\begin{proof}
Denote $B_{\ell}(\bold{q}) = \mathbbm{1}\{\forall r \in N_L(\ell),~q_r = b\}$. That is, whether all neighbors of port $\ell$ are full. Then by definition,
\[
\begin{aligned}
p_{\set{B}} &= \frac{1}{\lambda_{\Sigma}}\sum_{\ell=1}^L \lambda_{\ell}\expect{B_{\ell}(\bar{\bold{Q}})} \\
&= \frac{1}{\lambda_{\Sigma}}\sum_{\ell=1}^L \lambda_{\ell}\expect{B_{\ell}(\bar{\bold{Q}})\middle| \sum_{m=1}^K C_m(\bar{\bold{Q}}) \leq 3}\Pr\left\{\sum_{m=1}^K C_m(\bar{\bold{Q}}) \leq 3\right\} \\
&\mspace{20mu} + \frac{1}{\lambda_{\Sigma}}\sum_{\ell=1}^L \lambda_{\ell}\expect{B_{\ell}(\bar{\bold{Q}})\middle| \sum_{m=1}^K C_m(\bar{\bold{Q}}) > 3}\Pr\left\{\sum_{m=1}^K C_m(\bar{\bold{Q}}) > 3\right\} \\
&\leq \frac{1}{\lambda_{\Sigma}}\sum_{\ell=1}^L \lambda_{\ell}\expect{B_{\ell}(\bar{\bold{Q}})\middle| \sum_{m=1}^K C_m(\bar{\bold{Q}}) \leq 3} + \Pr\left\{\sum_{m=1}^K C_m(\bar{\bold{Q}}) > 3\right\}.
\end{aligned}
\]
To bound $\Pr\left\{\sum_{m=1}^K C_m(\bar{\bold{Q}}) > 3\right\}$, notice that $C^* \leq 1$, so 
\[
\Pr\left\{\sum_{m=1}^K C_m(\bar{\bold{Q}}) > 3\right\} \leq \Pr\left\{\sum_{m=1}^K C_m(\bar{\bold{Q}}) > C^* + 2\right\} \leq \frac{52\tau_{1K}b^2}{\epsilon N}
\]
by Lemma \ref{lemma:prob-serverK-explode}.

Then for the case $\sum_{m=1}^K C_m(\bold{q}) \leq 3$, it holds that $\sum_{m=1}^K s_{m,b}(\bold{q}) \leq \frac{3}{b}$. Let $\set{I}$ be the set of servers of types no greater than $K$ with queue length less than $b$. Then we know $|\set{I}| \geq (1 - \frac{3}{b})N \geq \frac{\hat{\beta}}{2}N$ since $b \geq 6$. By Assumption \ref{as:dense}, the total arrival rate of ports not connected with $\set{I}$ is thus upper bounded by $N\tilde{d_2}$. As a result,
\[
p_{\set{B}} \leq \frac{1}{\lambda_{\Sigma}}\sum_{\ell=1}^L \lambda_{\ell}\expect{B_{\ell}(\bar{\bold{Q}})\middle| \sum_{m=1}^K C_m(\bar{\bold{Q}}) \leq 3} + \Pr\left\{\sum_{m=1}^K C_m(\bar{\bold{Q}}) > 3\right\} \leq \frac{\tilde{d}_2}{\lambda} + \frac{52\tau_{1K}b^2}{\epsilon N}.
\]
\end{proof}

\subsection{Proof of Corollary \ref{cor:asymptotic}}
\renewcommand{\thmnamerestated}{Corollary~\ref{cor:asymptotic}[Restated]}
\begin{thmrestatedfilling}
Suppose that $\epsilon_N$ is both $o(1)$ and $\omega(N^{-0.5}\ln(N))$, and that both Assumptions \ref{as:buffer} and \ref{as:dense} hold for $G_N$ when $N$ is sufficiently large. Then as $N \to \infty$, both JFSQ and JFIQ are asymptotically optimal, and the expected queueing delay converges to zero for both policies.
\end{thmrestatedfilling}
\begin{proof}
First since $\epsilon_N = \omega(\ln N N^{-0.5})$, there is always a $b_N$ satisfying Assumption \ref{as:buffer} when $N$ is sufficiently large. Let $\bar{\bold{Q}}_N$ be the queue-length random variable, and let $p^N_{\set{B}}$ be the blocking probability for the $N-$th system. Applying Theorem \ref{thm:finite-bound} gives 
\[
\expect{\sum_{m=1}^M C_m(\bar{\mathbf{Q}}_N)} \leq C^*+\left(1+\frac{\tau_{KM}}{2}\right)\epsilon_N+2\sqrt{\frac{5\tau_{1M}b_N\ln N}{ N}}+60b_N^2\sqrt{\frac{26\tau_{1K}\tau_{1M}}{\hat{\beta_N}\epsilon_N N}},
\]
and $p^N_{\set{B}} \leq \frac{\epsilon_N\mu_K}{2b_N\lambda}+\frac{52\tau_{1K}b_N^2}{\epsilon_N N}$ for $N$ large enough.

Since $\epsilon_N = o(1), \epsilon_N=\omega(N^{-0.5}\ln N), \hat{\beta_N} > \epsilon_N$ and $b_N$ satisfies Assumption \ref{as:buffer}, it holds that $\lim_{N \to \infty} \expect{\sum_{m=1}^M C_m(\bar{\mathbf{Q}}_N)} = C^*$. Then by Little's Law, the expected mean response time $\expect{T_N}$ of the $N-$th system is given by the mean number of jobs in the system divided by the effective arrival rate. Therefore, 
\[
\lim_{N \to \infty} \expect{T_N} = \lim_{N \to \infty} \frac{\expect{N\sum_{m=1}^M C_m(\bar{\bold{Q}}_N)}}{\lambda_{\Sigma}(1-p_{\set{B}}^N)} \leq \frac{C^*}{\lambda\left(1 - \lim_{N \to \infty} \frac{\epsilon_N\mu_K}{2b_N\lambda}+\frac{52\tau_{1K}b_N^2}{\epsilon_N N}\right)} = \frac{C^*}{\lambda},
\]
which matches the lower bound in Theorem \ref{thm:finite-bound}. Therefore, JFSQ and JFIQ are asymptotically optimal in mean response time. On the other hand, let $\expect{T_{\set{W}}^N}$ be the expected waiting time of jobs, and let $\expect{Z_N}$ be the expected service time in the $N-$th system. Then it holds $\expect{T_N} = \expect{T_{\set{W}}^N} + \expect{Z_N}$. Since $\expect{Z_N} \geq \frac{C^*}{\lambda}, \expect{T_{\set{W}}^N} 
\geq 0$, and $lim_{N \to \infty} \expect{T_N} = \frac{C^*}{\lambda}$, it holds $\lim_{N \to \infty} \expect{T_{\set{W}}^N} = 0$. As a result, JFSQ and JFIQ obtain asymptotic zero queueing delays. 
\end{proof}

\section{Proof of Random Graph Results}
Here we provide the missing proof of Theorem \ref{thm:random-graph-uniform}.
\subsection{Proof of Theorem \ref{thm:random-graph-uniform}}
\renewcommand{\thmnamerestated}{Theorem~\ref{thm:random-graph-uniform}[Restated]}
\begin{thmrestatedfilling}
Suppose that all ports share the same arrival rates, that is, $\lambda_{\ell}\equiv \bar{\lambda}$ for all $\ell \in \set{L}$. Then following the same construction of graph $G$ in Theorem \ref{thm:random-graph} but with $H_j = 6\left(-\ln{p_j}+\frac{\tilde{d_j}}{p_j\bar{\lambda}}\ln\frac{2\mu_1}{\tilde{d_j}}\right)$ for $j \in \{1,2\}$, it holds that $G$ satisfies Assumption \ref{as:dense} with probability at least $1 - 2\binom{N}{Np_1}^{-1}$. The total number of edges in $G_N$ scales as $O\left(\frac{(N+L)b^3}{\epsilon}\ln\frac{b}{\epsilon}\right)$.
\end{thmrestatedfilling}
\begin{proof}
The proof is similar to that of Theorem \ref{thm:random-graph}. Let us follow the same notation in the proof of Theorem \ref{thm:random-graph}. Fix $j \in \{1,2\}$. Similarly, let $\set{K}$ be any subset of $\set{L}$ satisfying $\sum_{\ell \in \set{K}} \lambda_{\ell} > N\tilde{d}_j$, and $\set{I}$ be any subset of $\set{R}^j$ satisfying $|\set{I}| \geq Np_j$. To bound $\Pr\{\set{D}_{\set{K},\set{I}}\}$, W.L.O.G., we can assume every port in $\set{K}$ has arrival rate less than $N\tilde{d_j}{H_j}$, otherwise $\Pr\{\set{D}_{\set{K},\set{I}}\} = 0$. Then following the same argument in the proof of Theorem \ref{thm:random-graph}, it holds
$
\Pr\{\set{D}_{\set{K},\set{I}}\} \leq \exp(-H_jNp_j).$ 

The key step is to obtain a bound on the number of pairs of feasible $\set{K},\set{I}$ so that we can use the union bound. Let $N_{\set{K}}^j, N_{\set{I}}^j$ be the amount of such sets, respectively. W.L.O.G., assume that $Np_j$ is an integer since $|\set{I}|$ must be an integer. Also, as all ports share the same arrival rate $\bar{\lambda}$, we can assume $N\tilde{d_j}/\bar{\lambda}$ is an integer since the size of $\set{K}$ must exceed this value. Then it holds that 

\begin{align}
N_{\set{K}}^j &= \binom{L}{N\tilde{d_j}/\bar{\lambda}} \leq \binom{\lceil N\mu_1 / \bar{\lambda} \rceil}{N\tilde{d_j}/\bar{\lambda}} \\
N_{\set{I}}^j &= \binom{N}{Np_j}.
\end{align}

We have the following lemma bounding a binomial number.
\begin{lemma}\label{lemma:bound-binomial}
Fix an integer $n$. For any $0 < \alpha < \frac{1}{2}$, if $\alpha n$ is an integer, then $\ln\left(\binom{n}{\alpha n}\right) \leq -3\alpha n\ln \alpha$.
\end{lemma}
\begin{proof}
Let $k = \alpha n$. It holds that 
\[
\binom{n}{k} = \frac{n(n-1)\cdots (n-k+1)}{k!} \leq \frac{n^k}{k!}.
\]
We know that $e^k = \sum_{i \geq 0} \frac{k^i}{i!}$. Therefore, $\frac{k^k}{k!} \leq e^k$. It then implies that 
\[
\binom{n}{k} \leq \frac{n^k}{k!} \leq \frac{e^k n^k}{k^k} = \left(\frac{en}{k}\right)^k.
\]
As a result,
\[
\ln\left(\binom{n}{\alpha n}\right) \leq \alpha n(1 - \ln(\alpha)) \leq -3n\alpha \ln \alpha
\]
because $\alpha < \frac{1}{2}$.
\end{proof}
Now by the definition of $p_j$, $\tilde{d}_j$, it holds $p_j < \frac{1}{2}, \frac{N\tilde{d}_j/\bar{\lambda}}{\lceil N\mu_1/\bar{\lambda} \rceil} < \frac{1}{2}.$ Then by Lemma \ref{lemma:bound-binomial}, when $N$ is sufficiently large, 
\begin{equation}
\ln\left(N_{\set{K}}^j\right) \leq -3Np_j\ln p_j,~~\ln\left(N_{\set{I}}^j\right) \leq -3N\tilde{d}_j/\bar{\lambda}\ln\left(\frac{2\mu_1}{\tilde{d}_1}\right).
\end{equation}
Therefore, it holds that 
\begin{equation}
\Pr\{\set{C}_j\} \leq N_{\set{K}}^j N_{\set{I}}^j \exp(-H_j N p_j) \leq \exp\left(-Np_jH_j-3Np_j\ln p_j - 3Np_j\frac{\tilde{d}_j}{p_j\bar{\lambda}}\ln\left(\frac{2\mu_1}{\tilde{d}_j}\right)\right).
\end{equation}
By definition, $H_j = 6\left(-\ln p_j - \frac{\tilde{d}_j}{p_j\bar{\lambda}}\ln\left(\frac{2\mu_1}{\tilde{d_j}}\right)\right)$. Then we can see
\[
\Pr\{\set{C}_j\} \leq \exp(3Np_j\ln p_j) \leq \binom{N}{Np_j}^{-1}.
\]
By the union bound, it holds that
\[
\Pr\{\set{C}_1 \cup \set{C}_2\} \leq 2\binom{N}{Np_1}^{-1}.
\]
since $p_1 < p_2 < \frac{1}{2}$. Therefore, the probability that $G_N$ satisfies Assumption \ref{as:dense} is at least $1 - 2\binom{N}{Np_1}^{-1}$.

For the total number of edges used in $G_N$, consider the four types of connections on graph $G_N$ as per Theorem \ref{thm:random-graph} and Theorem \ref{thm:random-graph-uniform} where we use different $H_j$. we bound the number of edges for each type as follows. First, through some calculations, $H_j = O\left(\left(1+\frac{1}{b\bar{\lambda}}\right)\ln\left(\frac{b}{\epsilon}\right)\right)$, and $\frac{H_j}{\tilde{d}_j} = O\left(\frac{b^3\bar{\lambda}+b^2}{\epsilon \bar{\lambda}}\ln\frac{b}{\epsilon}\right)$.

Then the number of ports with $\lambda_{\ell} \geq N\frac{\tilde{d}_1}{H_1}$ is bounded by $\frac{L\bar{\lambda}H_1}{N\tilde{d}_1} = O\left(\frac{(N+L)b^3}{N\epsilon}\ln\frac{b}{\epsilon}\right)$ because $\lambda_{\Sigma} = L\bar{\lambda}$. Therefore, the number of connections from them is bounded by $O\left(\frac{(N+L)b^3}{\epsilon}\ln\frac{b}{\epsilon}\right)$ since there are $N$ servers. The same result holds for ports with $\lambda_{\ell} \geq N\frac{\tilde{d}_2}{H_2}$. Now for the remaining ports, the expected number of edges is upper bounded by 
\[
2\sum_{\ell \in \set{L}}\frac{\lambda_{\ell}}{N}\left(\frac{H_1}{\tilde{d}_1}+\frac{H_2}{\tilde{d}_2}\right)N = O\left(\frac{(N+L)b^3}{\epsilon}\ln\frac{b}{\epsilon}\right).
\]
Then to sum up, the expected number of edges in $G_N$ scales as $O\left(\frac{(N+L)b^3}{\epsilon}\ln\frac{b}{\epsilon}\right)$.
\end{proof}

\section{Additional Simulation Results}
In this section, we provide missing details in the main text and give additional simulation results.
\subsection{Description of JSQ-(2,2)}
In JSQ-(2,2)\cite{gardner2020scalable}, there are two parameters $p_F,p_S$. Then for each arrival of jobs, we find a server as follows:
\begin{enumerate}
\item sample $2$ fast servers and $2$ slow servers;
\item if there is an idle fast server, route the job to this server;
\item if there is an idle slow server, route the job to this server with probability $p_S$, and route the job to the fast server with shorter queue with probability $1 - p_S$;
\item otherwise, route the job to the fast server with shorter queue with probability $p_F$; and route the job to the slow server with shorter queue with probability $p_S$.
\end{enumerate}
We set $p_S,p_F$ to be the optimal values from Table 1 in \cite{gardner2020scalable}.
\subsection{Convergence of Blocking Probability}
Fig. \ref{fig:block-random} provides the convergence of the blocking probability following the same setting as in Section \ref{sec:random-manyserver}.
\begin{figure}
	\centering
	\scalebox{0.45}{\input{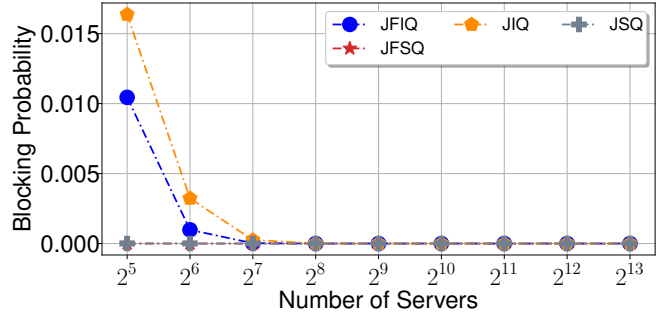}}
	\caption{The Blocking Probability of Different Routing Policies on Increasing-Sized Random Bipartite Graphs}
	\label{fig:block-random}
\end{figure}
Unlike JSQ which is shown to be throughput optimal \cite{cruise2020stability} (so is JFSQ), JIQ and JFIQ could lose the capacity of the system. As in Fig. \ref{fig:block-random}, when we set the buffer size to be $5$, the blocking probability of JIQ is around 1.5 percent, and that of JFIQ is around 1 percent. Interestingly, JFIQ seems to be more stable. Nevertheless, the blocking probability of both algorithms decreases swiftly as $N$ increases.
\subsection{Exploring More General Service Time Distribution}
We present a preliminary study here that extends results proved in this paper. Roughly speaking, we consider the same setting as in Section \ref{sec:random-manyserver}. However, we allow the service time distribution to be hyper-exponential.  

Still, suppose there are $N$ servers in the system where $N$ can scale up. Servers can be classified into four types with different service speed. Each type consists of the same amount of servers. Then let $X$ be a hyper-exponential distribution such that $X\sim \mathrm{Exp}(0.01)$ with probability $0.01$, and $X \sim \mathrm{Exp}(1)$ with probability $0.99$. The coefficient of variation of $X$ is around $7.071$, which is higher than that of an exponential distribution. Then for a type $i$ servers with $i \in \{1,2,3,4\}$, we assume that the service time of a job at this server is independently and identically distributed as $2^{i-1}X.$ Similarly, we can define the service rate of type-$i$ servers as $\mu_i = \frac{1}{2^{i-1}\expect{X}}$. Then the system load is defined as $\frac{4\lambda_{\Sigma}}{\sum_{i=1}^4 N\mu_i}$ where $\lambda_{\Sigma}$ is the total arrival rate. We can also obtain the lower bound of the mean response time as in Proposition \ref{thm:lower-bound}.

The buffer size is set as $b = 5$.  Following the same setting of ports and construction of the random graph, we obtain Fig. \ref{fig:hyperexp-random} for the mean response time of different policies, and the blocking probability is shown in Fig.\ref{fig:hyperexp-block}.
\begin{figure}
	\centering
	\scalebox{0.45}{\input{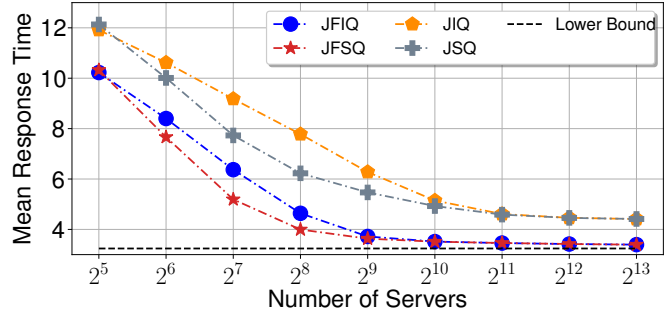}}
	\caption{The Mean Response Time of Different Routing Policies when Service Time is Hyper-Exponential}
	\label{fig:hyperexp-random}
\end{figure}
\begin{figure}
	\centering
	\scalebox{0.45}{\input{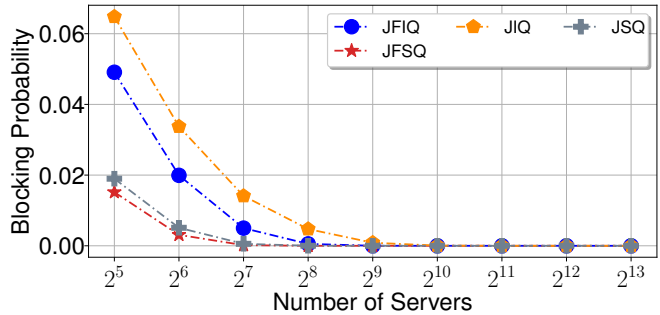}}
	\caption{The Blocking Probability of Different Routing Policies when Service Time is Hyper-Exponential}
	\label{fig:hyperexp-block}
\end{figure}
Notice that the performance of each policy degrades a lot for small systems compared with Fig. \ref{fig:random}. But when the system size scales up, both JFSQ and JFIQ have favorable mean response time, which is very close to the lower bound. It suggests that our theoretical results may hold for general distributions, which we leave for future studies.
\end{document}